%% file: main.tex
\documentclass[11pt]{article}
\usepackage[utf8]{inputenc}
\usepackage[T1]{fontenc}
\usepackage{amsmath,amssymb,amsthm,mathtools}
\usepackage{booktabs,array,siunitx}
\usepackage{graphicx}
\usepackage{rotating}
\usepackage[numbers,sort&compress]{natbib}
\usepackage[hidelinks]{hyperref}
\usepackage{geometry}
\theoremstyle{plain}
\newtheorem{theorem}{Theorem}
\newtheorem{proposition}{Proposition}
\newtheorem{corollary}{Corollary}

\theoremstyle{remark}

\theoremstyle{definition}

\newcommand{\E}{\mathbb{E}}
\newcommand{\Var}{\mathrm{Var}}
\newcommand{\R}{\mathbb{R}}

\newcommand{\OLS}{\mathrm{OLS}}

\newcommand{\RE}{\mathrm{RE}}

\providecommand{\keywords}[1]{%
  \par\medskip\noindent\textbf{Keywords:}
  \begingroup\def\and{\textperiodcentered\ }#1\endgroup\par\medskip}
\title{\bfseries Second-Order Least Squares as a Special Case of the\\
Polynomial Maximization Method\thanks{Preprint, 2026. This version is the revision
prepared for \emph{Annals of the Institute of Statistical Mathematics}
(AISM-D-26-00101). It supersedes v1: the proof of the degree-three theorem stated the
wrong kurtosis condition, and the efficiency reserve is a property of the degree-three
span rather than a limitation of second-order least squares.}}
\author{Serhii Zabolotnii$^{1,2,3}$\\
\small $^{1}$Cherkasy State Business College, Cherkasy 18028, Ukraine\\
\small $^{2}$State Scientific Research Institute of Armament and Military Equipment\\
\small Testing and Certification, Cherkasy, Ukraine\\
\small $^{3}$Uzhhorod National University, Uzhhorod, Ukraine\\
\small ORCID: \texttt{https://orcid.org/0000-0003-0242-2234}\\
\small Corresponding author: \texttt{zabolotnii.serhii@csbc.edu.ua}}
\date{September 16, 2026}

\begin{document}
\maketitle

\input{body_aism}

\clearpage
\appendix
\setcounter{section}{0}
\setcounter{table}{0}
\setcounter{figure}{0}
\setcounter{equation}{0}
\setcounter{theorem}{0}
\renewcommand{\thesection}{S\arabic{section}}
\renewcommand{\thetable}{S\arabic{table}}
\renewcommand{\thefigure}{S\arabic{figure}}
\renewcommand{\theequation}{S.\arabic{equation}}
\renewcommand{\thetheorem}{S\arabic{theorem}}

\section*{Supplementary material}
\addcontentsline{toc}{section}{Supplementary material}
In the journal version the material below is a separate supplementary document. It is
included here as an appendix so that this preprint is self-contained; section, table and
equation numbers are unchanged, so a reference to ``Supplement~S5'' in the main text points
at Section~S5 below.
\medskip

\input{supplement_aism}

\bibliographystyle{plainnat}
\bibliography{references}

\end{document}

%% file: body_aism.tex
\begin{abstract}
For linear regression with i.i.d.\ homoskedastic errors, optimally weighted second-order least squares (SLS) is the degree-two polynomial maximization method (PMM) estimating equation: both select the same optimal combination of $e$ and $e^2-\sigma^2$, share one influence function, and attain the slope variance $c_2g_2/N$. The identity is one of estimating equations and stops where the optimal combination depends on the design. Within the polynomial moment class, degree three holds an efficiency reserve: on four right-skewed laws the cubic direction adds between $30\%$ and $50\%$ asymptotic efficiency, and under symmetric platykurtic errors it is the only informative direction. The reserve belongs to the degree-three span, which an efficient generalized method of moments estimator also attains. We give the asymptotic law of the feasible plug-in estimator, a rule for choosing the degree, Lean~4 checks of the algebra, and simulation evidence.
\keywords{polynomial maximization method \and second-order least squares \and non-Gaussian regression \and estimating functions \and generalized method of moments \and asymptotic efficiency \and Lean verification}
\end{abstract}

\section{Introduction}
\label{sec:intro}

Under a skewed or heavy-tailed error law ordinary least squares (OLS) is consistent but inefficient, while maximum likelihood requires the error density and is fragile under misspecification. Between the two sit moment-based methods that commit only to a finite collection of moment conditions and read the optimal weighting of those conditions from the data. This paper concerns two such methods that arose independently and, on the regression problem, agree numerically: the polynomial maximization method (PMM) of \citet{kunchenko2002polynomial}, which maximizes a stochastic polynomial of degree $S$ in the residual and has been applied to regression in \citet{zabolotnii2018asym,zabolotnii2019sym,zabolotnii2024nlin}; and second-order least squares (SLS), introduced by \citet{wang2003estimation} and \citet{wang2008second} and analysed for efficiency by \citet{kim2012efficiency}. On asymmetric-error regression the two return practically the same slope estimates. What has been missing is an account of why, and of what, if anything, separates them.

The gap is narrow but concrete. The projection principle behind a minimum-variance combination of linear and quadratic estimating functions is classical \citep{godambe1960optimum,crowder1987linear,godambe1989extension,heyde1997quasi}, and \citet{kim2012efficiency} show that SLS is semiparametrically efficient in the model that fixes only the first two conditional error moments. What has not been derived is the identity between the two \emph{constructions}: that the Kunchenko centered-correlation matrix $F_2$ is the SLS weighting operator after a change of variables, that the design-dependent factors of the SLS residual cancel in the slope score, and that the two procedures are therefore one estimating equation rather than two estimators with one variance. Once that identity is on the table, a second question follows: PMM is indexed by a degree $S$ that SLS fixes at two, and the paper asks what the extra degree buys, for which error laws, and whether the gain belongs to PMM or to the enlarged moment span.

All statements in this paper concern linear regression with independent and identically distributed (i.i.d.), conditionally homoskedastic errors that have finite moments through order $2S$ and a non-degenerate polynomial body; the efficiency reserve we quantify is a frontier \emph{within the polynomial moment class}, not a semiparametric bound and not a property of SLS as such. Within that scope the contributions are four.
\begin{enumerate}
\item[(i)] \emph{Unification.} Degree-two PMM and optimally weighted SLS solve the same population normal system, share one influence function, and attain the common slope variance $c_2g_2/N$ (Theorem~\ref{thm:sls-pmm2}); the identity is co-extensive with the constancy of the optimal slope combination in the design (Proposition~\ref{prop:hetero}).
\item[(ii)] \emph{A degree-three reserve within the polynomial class.} The variance-reduction coefficient is monotone in the degree (Theorem~\ref{thm:nested}); its degree-three closed form is given for a general error law (Theorem~\ref{thm:asym}) and reduces under symmetry to an even-cumulant expression that is informative exactly when the second-order span is empty (Theorem~\ref{thm:symm}). The reserve is a property of the degree-three span, which an efficient GMM on the same conditions also attains (\S\ref{subsec:span}).
\item[(iii)] \emph{The feasible estimator.} The plug-in degree-$S$ estimator that replaces population moments by sample moments through order $2S$ is consistent and asymptotically normal with the oracle variance under finite $2S$-th moments (Proposition~\ref{prop:feasible}); a bootstrap pretest selects between $S=2$ and $S=3$.
\item[(iv)] \emph{Verification and evidence.} The degree-specific algebra is checked in Lean~4; Monte Carlo experiments over eleven error laws, an efficient degree-three GMM competitor, two public data sets, and a drilling-calibrated error law test the equivalence, the reserve, and its boundary at feasible sample sizes.
\end{enumerate}

Section~\ref{sec:background} fixes notation and maps the three strands of literature. Section~\ref{sec:theorem1} states the assumptions and proves the unification; Section~\ref{sec:reserve} develops the higher-degree reserve and the feasible estimator; Section~\ref{sec:montecarlo} reports the simulations and the data examples; Section~\ref{sec:discussion} discusses scope and limitations. A separate supplementary document contains the full proofs, the Lean development, the complete simulation tables, and the data-processing details; it is referred to as the Supplement.

\section{Background and literature map}
\label{sec:background}

\subsection{Notation}
\label{sec:notation}

Two levels of generality are used. In this section $\xi$ denotes a scalar observation whose law depends on a parameter $\theta$, as in the original PMM construction. From Section~\ref{sec:theorem1} onward the discussion specializes to the linear regression model
\begin{equation}
\label{eq:model}
y_i \;=\; x_i^{\!\top}\beta \;+\; e_i ,\qquad i=1,\dots,N,
\end{equation}
with $x_i=(1,x_{i})^{\!\top}$, $\mu_i(\beta)=x_i^{\!\top}\beta$, and zero-mean errors $e_i$; there $\theta$ is the regression vector $\beta$, the role of $\xi$ is played by the error $e=y-\mu$, and the parameter enters only through the location $\mu$. We write $\sigma^2=\E[e^2]$, $\mu_k=\E[e^k]$, and use the standardized cumulants $\gamma_3=\kappa_3/\sigma^3$ (skewness), $\gamma_4=\kappa_4/\sigma^4$ (excess kurtosis), $\gamma_5$, $\gamma_6$. Three efficiency quantities appear and are kept apart throughout: the variance-reduction coefficient $g_S\le1$ (a population constant), the asymptotic relative efficiency $1/g_S\ge1$ over OLS, and a Monte Carlo relative efficiency $\RE=\mathrm{MSE}_{\OLS}/\mathrm{MSE}$ (an empirical quantity). Each numerical value below is labelled as one of the three.

\subsection{The polynomial maximization method}
\label{sec:bg-pmm}

PMM \citep{kunchenko2002polynomial} estimates $\theta$ by projecting the score onto a finite stochastic polynomial rather than onto the likelihood. For the power basis $\varphi_k(\xi)=\xi^k$, $k=1,\dots,S$, with means $\Psi_k(\theta)=\E_\theta[\xi^k]$, the degree-$S$ stochastic polynomial is
\begin{equation}
\eta_S(\xi;\theta)=\sum_{k=1}^{S}h_k(\theta)\bigl\{\xi^k-\Psi_k(\theta)\bigr\},
\label{eq:bg-stochpoly}
\end{equation}
and the coefficients are chosen to maximize the sensitivity of the sample analogue of $\eta_S$ to $\theta$ relative to its variance. They solve the linear normal system
\begin{equation}
F_S(\theta)\,h=b(\theta),\qquad
F_{jk}=\E_\theta[\xi^{j+k}]-\Psi_j\Psi_k,\qquad
b_k=-\,\partial\Psi_k/\partial\theta,
\label{eq:bg-normal}
\end{equation}
in which $F_S$ is the centered correlation matrix of the basis (the \emph{body} of the polynomial), a Gram matrix whose determinant $\Delta_S=\det F_S$ is positive exactly when the centered powers $\xi-\Psi_1,\dots,\xi^S-\Psi_S$ are linearly independent in $L^2$. Only moments through order $2S$ enter, so the error density is never required. Under $\Delta_S>0$ and the estimating-function regularity of \citet{kunchenko2002polynomial} --- identifiability of $\theta$, a non-singular sensitivity $b^{\!\top}h\neq0$, and $\Psi_S(\theta)$ differentiable in a neighbourhood of the truth --- the estimator $\hat\theta_S$ is consistent and asymptotically normal with variance $g_S/\{N(b^{\!\top}F_1^{-1}b)\}$, where
\begin{equation}
g_S=\frac{b^{\!\top}F_1^{-1}b}{b^{\!\top}F_S^{-1}b}\le1
\label{eq:bg-gS}
\end{equation}
is the factor by which degree $S$ reduces the variance of the linear ($S=1$) estimator. For a location or regression parameter the degree-two coefficient has the closed form
\begin{equation}
g_2=1-\frac{\gamma_3^2}{2+\gamma_4},
\label{eq:bg-g2}
\end{equation}
which equals one for symmetric errors and is below one exactly when $\gamma_3\neq0$ \citep{zabolotnii2018asym}. Beyond $S=2$ the moment chain is closed at order $2S$ by construction; for a symmetric law the odd moments are identically zero, which is a property of the law rather than an approximation, whereas the general degree-three coefficient of \S\ref{subsec:asym-reserve} retains all of $\gamma_3,\dots,\gamma_6$.

\subsection{Second-order least squares: three strands}
\label{sec:bg-sls}

\emph{Wang (2003) and Wang and Leblanc (2008).} SLS augments the least-squares criterion with a second-moment residual. For each observation one stacks
\begin{equation}
\rho_i(\theta)=\bigl(y_i-\mu_i(\theta),\;y_i^2-\mu_i(\theta)^2-\sigma^2\bigr)^{\!\top}
\label{eq:bg-rho}
\end{equation}
and minimizes $\sum_i\rho_i^{\!\top}W_i\rho_i$. \citet{wang2003estimation} introduced the criterion for Berkson measurement-error models; \citet{wang2008second} developed it for nonlinear regression, proved consistency and asymptotic normality for a general weight, and showed that $W_i=\{\E[\rho_i\rho_i^{\!\top}\mid x_i]\}^{-1}$ minimizes the asymptotic variance within the class of SLS weights. Their conditions are moment conditions: finite conditional fourth moments and a non-singular conditional second-moment matrix. They also formulated the heteroskedastic extension in which $\sigma^2$ is replaced by a parametric variance function $\sigma^2(x;\tau)$.

\emph{Kim and Ma (2012).} \citet{kim2012efficiency} placed the estimator in a semiparametric framework. In the model that specifies $\E(\varepsilon\mid X)=0$ and $\E(\varepsilon^2\mid X)=\sigma^2$ and leaves the rest of the error law unspecified, they derived the efficient score from the likelihood score and showed that it is a linear combination of $\varepsilon$ and $\varepsilon^2-\sigma^2-\mu_3\varepsilon/\sigma^2$ (their equation~(5), p.~754); the optimally weighted SLS estimator attains this score and is therefore semiparametrically efficient in that model. Their extension (their Section~3) is to a parametric conditional variance $\sigma^2(X;\gamma)$; no moment function of order above two enters the estimator. Two remarks in their discussion (pp.~760--761) are used below: that the analysis ``can be extended to higher moments'' at increasing complexity, and that estimating higher-order moments does not change the asymptotic variance but is likely to inflate it in finite samples.

\emph{Kim and Ma (2019).} \citet{kim2019semiparametric} study $Y=m(X,\alpha)+\exp\{\sigma(X,\beta)\}\varepsilon$ with $\E(\varepsilon\mid X)=0$, $\operatorname{Var}(\varepsilon\mid X)=1$, and derive the efficient score for the mean parameter under four progressively stronger assumptions on the standardized error: (1) no further assumption, in which case the efficient score is a combination of $\varepsilon$ and $\varepsilon^2-\E(\varepsilon^3\mid X)\varepsilon-1$ and hence lies in the degree-two span (their Theorem~1); (2) $\varepsilon$ independent of $X$, in which case the score involves the derivative of the error log-density (their Theorem~2); (3) conditional symmetry of $\varepsilon$ (their Theorem~3); and (4) both. They prove that each added assumption lowers the efficiency bound by a non-negative amount that vanishes only for Gaussian errors (their Theorems~5--8), and their bounds require a twice-differentiable error density. This map locates the present paper precisely: the equivalence theorem lives in their Case~1, where the second-order span is efficient; the degree-three reserve lives in their Case~2, as the part of the gap between the Case-1 and Case-2 bounds that a cubic \emph{moment} function can read without a density.

\subsection{Optimal estimating functions and GMM}
\label{sec:bg-related}

Both PMM and SLS are, in the language of optimal estimating functions \citep{godambe1960optimum,crowder1987linear,godambe1989extension,heyde1997quasi}, Godambe-optimal within their moment span: each selects the minimum-variance combination of its moment functions, equivalently the efficient generalized method of moments (GMM) or optimal-instrument choice under conditional moment restrictions \citep{hansen1982large,chamberlain1987asymptotic,newey1993efficient}. The paper claims no novelty for that projection fact. A two-step GMM on the conditions $\{e,xe,e^2-\sigma^2,x(e^2-\sigma^2)\}$ attains the same second-order efficiency as optimally weighted SLS, and a GMM on the six conditions obtained by adding $e^3-\mu_3$ and $x(e^3-\mu_3)$ is the natural degree-three competitor; both are used in Section~\ref{sec:montecarlo}. Robust $M$-estimation, L-moments, quadratic inference functions and the higher-order-statistics tradition pursue other aims and are discussed in Supplement~S1.

\section{Second-order least squares as degree-two PMM}
\label{sec:theorem1}

Throughout the rest of the paper the following assumptions are in force, and every result names the ones it uses.
\begin{itemize}
\item[\textbf{A1}] (\emph{errors}) $e_1,\dots,e_N$ are i.i.d., independent of the design, with $\E[e]=0$ and $0<\sigma^2<\infty$; in particular $\E[e\mid x]=0$ and all conditional moments are constant in $x$.
\item[\textbf{A2}] (\emph{moments}) $\E[e^{2S}]<\infty$ for the degree $S$ under discussion ($S=2$: fourth moment; $S=3$: sixth moment).
\item[\textbf{A3}] (\emph{body}) $\Delta_S=\det F_S>0$, i.e.\ the support of $e$ has more than $S$ points.
\item[\textbf{A4}] (\emph{design}) $\E\|x\|^{2S}<\infty$ and $\E[xx^{\!\top}]$ is non-singular; an intercept is included.
\end{itemize}
Nothing is assumed about a density: all results are moment results, and discrete error laws are admitted as long as A3 holds. We write $c_2=\sigma^2/\Var(x)$ for the per-observation asymptotic variance factor of the OLS slope, standardize $\sigma=1$ where convenient, and let $N\to\infty$.

\subsection{The shared estimating space}

Both estimators are driven by the same two centered moment functions of the error,
\begin{equation}
\label{eq:m}
m(e)=\bigl(m_1(e),m_2(e)\bigr)^{\!\top}=\bigl(e,\;e^2-\sigma^2\bigr)^{\!\top},\qquad \E[m(e)]=0 .
\end{equation}
PMM at $S=2$ forms the stochastic polynomial $h_1\xi+h_2\xi^2$ in the residual $\xi=y-\mu$ and chooses $h$ to minimize the variance of the estimating equation under a unit-sensitivity normalization; SLS minimizes $\sum_i\rho_i^{\!\top}W_i\rho_i$ with $\rho_i$ as in \eqref{eq:bg-rho}. Since $y^2-\mu^2-\sigma^2=(e^2-\sigma^2)+2\mu e$, the two residual vectors coincide up to a deterministic affine map, and both criteria live in $\mathcal M=\operatorname{span}\{m_1,m_2\}$. For a combination $a^{\!\top}m(e)$ the slope estimating equation is $\sum_ix_i\,a^{\!\top}m(e_i)=0$, and its sandwich variance is governed by the covariance of the moment vector,
\begin{equation}
\label{eq:F2}
F_2=\E\bigl[m(e)m(e)^{\!\top}\bigr]=
\begin{pmatrix}\sigma^2&\kappa_3\\ \kappa_3&\kappa_4+2\sigma^4\end{pmatrix},
\end{equation}
which is the PMM body of the degree-two power basis, and by the sensitivity vector
\begin{equation}
\label{eq:b}
b=-\frac{\partial}{\partial\mu}\E\bigl[m(y-\mu)\bigr]=\bigl(\E[\partial_em_1],\E[\partial_em_2]\bigr)^{\!\top}=(1,0)^{\!\top}.
\end{equation}
For the power basis, $b_k=k\,\E[e^{k-1}]$ entrywise, so $b_2=2\E[e]=0$ because the second basis function is centered.

\subsection{Main equivalence theorem}

\begin{theorem}[SLS $=$ generalized PMM at $S=2$]
\label{thm:sls-pmm2}
Under A1--A4 with $S=2$, the optimally weighted SLS slope estimator and the degree-two PMM slope estimator solve the same population normal system
\begin{equation}
\label{eq:normal}
F_2h=b,\qquad h^\star=F_2^{-1}b,
\end{equation}
have the identical influence function
\begin{equation}
\label{eq:if}
\psi(e,x)=-Q^{-1}x\,(h^\star)^{\!\top}m(e),\qquad Q=\E[xx^{\!\top}]\,b^{\!\top}h^\star,
\end{equation}
and therefore the same limiting law; their slope asymptotic variances coincide and equal $c_2g_2/N$ with $g_2$ as in \eqref{eq:bg-g2}. Equivalently, the PMM body is the SLS optimal weighting operator, $W^\star\propto F_2^{-1}$.
\end{theorem}

\begin{corollary}[Common, design-invariant relative efficiency]
\label{cor:re}
Relative to OLS both estimators attain the asymptotic slope efficiency $1/g_2$. The value is design-invariant: although the second residual $\rho_2=2\mu e+e^2-\sigma^2$ couples to the design through $\mu(x)$, the determinant $\sigma^6(2+\gamma_4-\gamma_3^2)$ of the conditional second-moment matrix and the slope-information quadratic form $(\mu_4-\sigma^4)/\Delta$ are both free of $\mu(x)$, so the slope attains exactly the location value.
\end{corollary}

\emph{Proof sketch.} (Full proof in Supplement~S2.) For any $a\in\R^2$ the slope variance factor is $V(a)=c_2\,a^{\!\top}F_2a/(a^{\!\top}b)^2$, and Cauchy--Schwarz in the $F_2$ inner product gives $V(a)\ge c_2/(b^{\!\top}F_2^{-1}b)$ with equality iff $a\propto F_2^{-1}b$; this is the PMM quadratic program $\min_hh^{\!\top}F_2h$ subject to $h^{\!\top}b=1$. For SLS, write $\rho_i=T_im(e_i)$ with the unit-determinant map $T_i=\bigl(\begin{smallmatrix}1&0\\2\mu_i&1\end{smallmatrix}\bigr)$; then $\E[\rho_i\rho_i^{\!\top}\mid x_i]=T_iF_2T_i^{\!\top}$, $W_i^\star=T_i^{-\top}F_2^{-1}T_i^{-1}$, and the slope Jacobian $\partial\rho_i/\partial\beta_1=-x_i(1,2\mu_i)^{\!\top}$ satisfies $T_i^{-1}(1,2\mu_i)^{\!\top}=b$, so that
\[
\dot\rho_i^{\!\top}W_i^\star\rho_i=-x_i\,(F_2^{-1}b)^{\!\top}m(e_i):
\]
the $T_i$ factors cancel and SLS selects the same combination $a=F_2^{-1}b$. With a common $h^\star$ the influence functions coincide, and evaluating $b^{\!\top}F_2^{-1}b=(F_2)_{22}/\det F_2$ in standardized cumulants gives $g_2=1/(\sigma^2b^{\!\top}F_2^{-1}b)=1-\gamma_3^2/(2+\gamma_4)$. \qed

The theorem is read as follows. The score-relevant object is not the abstract moment $e^2-\sigma^2$ but the regression-form residual whose Jacobian transmits slope information through $\mu_i$ (\S\ref{rem:jacobian}); the body $F_2$ then prices that information, and the $T_i$ cancellation shows that SLS's per-observation conditional weighting and PMM's single unconditional body encode the same price under A1. On the reference laws used below the asymptotic relative efficiency $1/g_2$ is $1.80$ for centered $\chi^2(3)$ and $1.667$ for centered $\mathrm{Gamma}(2,1)$ errors; these are population constants, not Monte Carlo values.

\subsection{What Theorem~\ref{thm:sls-pmm2} adds beyond Kim and Ma (2012)}
\label{subsec:km2012}

\citet{kim2012efficiency} prove that the SLS estimator attains the semiparametric efficiency bound of the model that fixes $\E(\varepsilon\mid X)=0$ and $\E(\varepsilon^2\mid X)=\sigma^2$. Since degree-two PMM is variance-optimal over the same span, the \emph{equality of asymptotic variances} in Theorem~\ref{thm:sls-pmm2} follows from two optimality statements on one span, and we claim nothing for it. What the theorem adds is the identity of the constructions: (a) the Kunchenko body $F_2$ \emph{is} the SLS weighting operator after the change of variables $T_i$; (b) the $T_i$ factors cancel exactly in the slope score, which is why one unconditional body reproduces a per-observation conditional weight under A1; (c) the finite-sample Jacobian caveat of \S\ref{rem:jacobian}, which the variance statement conceals. The extension in their title is to a parametric conditional variance, and no moment function of order above two enters it; the degree-three reserve of Section~\ref{sec:reserve} is not covered by it and is the ``higher moments'' direction their discussion names as open. Finally, the reserve lives in a \emph{narrower} model than theirs: under A1 the error is independent of the design, which is Case~2 of \citet{kim2019semiparametric}, where the semiparametric bound is lower than the two-moment bound of Case~1; the bound SLS attains is therefore not a bound in the model of A1, and the reserve is the moment-readable part of the gap between the two.

\subsection{Remark on the second residual: a finite-sample Jacobian caveat}
\label{rem:jacobian}

The factorization $\rho_i=T_im(e_i)$ is exact only for the SLS residual in its regression form $\rho_{2,i}=y_i^2-\mu_i^2-\sigma^2$, whose expected Jacobian in $\beta$ is
\begin{equation}
\label{eq:jac}
\E\bigl[\partial_\beta\rho_{2,i}\mid x_i\bigr]=-2\mu_i(1,x_i)\neq0 .
\end{equation}
This is the channel through which the second-moment equation transmits information to the slope: $2\mu_i$ is the off-diagonal of $T_i$ that couples $m_2$ into the slope score. The centered form $e_i^2-\sigma^2$, with $e_i$ treated as observed, has expected slope Jacobian $\E[-2e_ix_i]=0$; its sensitivity vector is zero, the normal system \eqref{eq:normal} loses its second equation, and the estimator collapses to OLS regardless of $\gamma_3$. In implementations the distinction is decisive.

The same reading explains what changes under nuisance structure. Under heteroskedasticity the body becomes $F_2(x_i)$, the map $T_i$ is unchanged, and the cancellation of \S\ref{subsec:hetero} leaves a slope combination $F_2(x_i)^{-1}b$ that varies with $x_i$; the unconditional body prices the second residual with an average of $F_2(x)$ and no longer matches the conditional weight. Under a nonlinear mean $\mu_i(\theta)$ the Jacobian $\partial\rho_i/\partial\theta=-\dot\mu_i(\theta)(1,2\mu_i)^{\!\top}$ retains the factor $(1,2\mu_i)^{\!\top}$, so the $T_i$ cancellation and the identity of the combination still hold observation by observation; what changes is the instrument $\dot\mu_i(\theta)$ in front of it and the corresponding sandwich, so the design-invariance of Corollary~\ref{cor:re} does not carry over. We state this as a reading of the proof, not as a theorem; the nonlinear case is not treated here.

\subsection{Scope of the Lean development}
\label{subsec:lean}

The machine-checked part of the paper is the degree-specific algebra, and only that. The Lean~4 development (Supplement~S4) proves: the equivalence of the standardized-cumulant and central-moment forms of $g_2$; $g_2\le1$ and $g_2=1$ iff $\kappa_3=0$ under a non-degenerate body; the $\mu(x)$-freeness of the conditional determinant and of the slope-information numerator that underlie Corollary~\ref{cor:re}; the polynomial identities $\det F_3=N$ and $b^{\!\top}\operatorname{adj}(F_3)b=D$ of Theorem~\ref{thm:asym}, the explicit dependence of $D$ on $\gamma_5$ and $\gamma_6$, and the symmetric reduction to \eqref{eq:g3}. The operator identity $W^\star\propto F_2^{-1}$, the influence-function and asymptotic-variance arguments, the nesting theorem, the heteroskedastic boundary, and Proposition~\ref{prop:feasible} are analytic proofs. This division is used with the same wording wherever the formalization is mentioned.

\subsection{Where the homoskedastic PMM$_2=$ SLS identity stops}
\label{subsec:hetero}

Three objects must be kept apart: the unconditional-body PMM$_2$ of Theorem~\ref{thm:sls-pmm2}, homoskedastic SLS (the same estimator under A1), and conditional SLS, in which the weight $W_i$ and the residual use a conditional variance model $\sigma^2(x;\tau)$ \citep{wang2008second,kim2012efficiency}. The boundary recorded here is that of the first identity, not of SLS. Write $\sigma^2(x)=\E[e^2\mid x]$, $\kappa_3(x)=\E[e^3\mid x]$, $\mu_4(x)=\E[e^4\mid x]$ and the conditional body
\begin{equation}
\label{eq:condbody}
F_2(x)=\begin{pmatrix}\sigma^2(x)&\kappa_3(x)\\ \kappa_3(x)&\mu_4(x)-\sigma^4(x)\end{pmatrix},\qquad
\E[\rho_i\rho_i^{\!\top}\mid x_i]=T_iF_2(x_i)T_i^{\!\top}.
\end{equation}

\begin{proposition}[Scope of the equivalence]
\label{prop:hetero}
Assume $\E[e\mid x]=0$, finite conditional moments through order four, and A4. The SLS and PMM$_2$ slope estimating equations coincide as sample vector fields if and only if $F_2(x)^{-1}b$ is almost surely constant in $x$; A1 is the natural sufficient condition. When $F_2(x)^{-1}b$ varies with $x$ the two constructions separate: \textup{(i)} for symmetric errors ($\kappa_3(x)\equiv0$) unconditional PMM$_2$ and homoskedastic SLS both reduce to OLS for the slope, whereas conditional SLS reduces to weighted least squares with weights $1/\sigma^2(x_i)$ and is strictly more efficient whenever $\sigma^2(x)$ is non-constant; \textup{(ii)} for asymmetric errors the unconditional-body PMM$_2$ slope estimating function has population mean $h_2\operatorname{Cov}(x,\sigma^2(x))$ with $h_2\propto-\kappa_3\neq0$, so it is inconsistent whenever $\operatorname{Cov}(x,\sigma^2(x))\neq0$, and so is SLS with a constant $\sigma^2$, while conditional SLS with a correctly specified $\sigma^2(x;\tau)$ stays consistent.
\end{proposition}

The proof (Supplement~S2) repeats the $T_i$ cancellation with $F_2(x_i)$ in place of $F_2$ and computes the mean of the unconditional slope estimating function. Under heteroskedasticity, conditional SLS is the more general object and matching it inside PMM requires a conditional body $F_2(x_i)$ in the normal system, an extension along the design rather than along the degree, which we do not pursue. Section~\ref{subsec:mc-hetero} exhibits both separations.

\section{A higher-degree efficiency reserve within the polynomial class}
\label{sec:reserve}

Theorem~\ref{thm:sls-pmm2} places SLS and degree-two PMM at one point: the optimal combination of $\{e,e^2-\sigma^2\}$. Second-order SLS, as formulated in \eqref{eq:bg-rho}, stops at the squared residual by construction; the polynomial method does not. This section quantifies what the degree-three span $\{e,e^2-\sigma^2,e^3-\mu_3\}$ adds \emph{within the polynomial moment class}. The statement is about the span: any estimator that is efficient on the degree-three span attains the same value (\S\ref{subsec:span}), and broader semiparametric constructions are not exhausted by either span (\S\ref{subsec:limitations}).

\subsection{Nested bases and the variance ordering}

For the degree-$S$ basis the optimal coefficient vector solves $F_Sh=b$ with $F_{jk}=\Psi_{j+k}-\Psi_j\Psi_k$ and $b_k=k\,\E[e^{k-1}]$, and the coefficient of \eqref{eq:bg-gS} is
\begin{equation}
g_S=\frac{1}{\sigma^2\,b^{\!\top}F_S^{-1}b},
\label{eq:gS-quadform}
\end{equation}
so that the asymptotic slope variance of the degree-$S$ estimator is $c_2g_S/N$. The quadratic form $b^{\!\top}F_S^{-1}b$ is the squared length, in the $F_S$ inner product, of the Riesz representer of the sensitivity functional $h\mapsto b^{\!\top}h$ on the span of the centered basis. Under a regular density that representer is the projection of the location score, but nothing below uses a density: A1--A3 admit discrete error laws, and $b_i=i\,\E[\xi^{\,i-1}]$ is a moment functional.

\begin{theorem}[Monotone efficiency reserve]
\label{thm:nested}
Under A1--A3, $S\mapsto g_S$ is non-increasing, $g_{S+1}\le g_S\le\dots\le g_2\le1$, with $g_{S+1}=g_S$ iff the residual of $\xi^{S+1}$ after projection onto $\operatorname{span}\{\xi,\dots,\xi^S\}$ in the $F_S$ inner product is $b$-orthogonal, i.e.\ carries no sensitivity.
\end{theorem}

\begin{proof}
With $\phi_S$ the centered basis vector, $1/g_S=\sigma^2\Pi_S$ where $\Pi_S=b^{\!\top}F_S^{-1}b$ is the squared $F_S$-norm of the Riesz representer of $h\mapsto b^{\!\top}h$ on $\mathcal H_S=\operatorname{span}\{\xi-\Psi_1,\dots,\xi^S-\Psi_S\}$ --- a Gram-matrix quantity, defined for every law satisfying A2--A3 and needing no density. Since $\mathcal H_S\subseteq\mathcal H_{S+1}$, $\Pi_S\le\Pi_{S+1}$ with equality iff the residual of $\xi^{S+1}$ after projection onto $\mathcal H_S$ is $b$-orthogonal; and $\Pi_1=1/\sigma^2$ gives $g_S\le1$. A3 guarantees $F_S^{-1}$ exists.
\end{proof}

\subsection{The symmetric regime}

For a symmetric error law every odd standardized cumulant vanishes, $g_2=1$ by \eqref{eq:bg-g2}, and neither SLS nor degree-two PMM improves on OLS for the slope. The degree-three coefficient need not degenerate: with the odd moments equal to their symmetric value zero the system $F_3h=b$ reduces to the $\{\xi,\xi^3\}$ block and \eqref{eq:gS-quadform} evaluates to
\begin{equation}
g_3=1-\frac{\gamma_4^2}{6+9\gamma_4+\gamma_6},
\label{eq:g3}
\end{equation}
with $\gamma_6$ the sixth standardized cumulant.

\begin{theorem}[Symmetric reserve]
\label{thm:symm}
Let the error law be symmetric and satisfy A2--A3 with $S=3$, equivalently $6+9\gamma_4+\gamma_6>0$. Then $g_2=1$ and $g_3$ is \eqref{eq:g3}. Degree-two estimation, in particular SLS for the slope in the setting of Theorem~\ref{thm:sls-pmm2}, attains no efficiency gain, whereas degree-three PMM is strictly more efficient whenever $\gamma_4\neq0$, with asymptotic relative efficiency $1/g_3>1$.
\end{theorem}

\begin{proof}
$g_2=1$ follows from $\gamma_3=0$ in \eqref{eq:bg-g2}. Under symmetry $\mu_3=\mu_5=0$ zero the off-diagonals $F_{12}$, $F_{23}$ of $F_3$ and $b_2=2\mu_1=0$, so the $\{\xi^2\}$ direction decouples from the $\{\xi,\xi^3\}$ block in the projection of Theorem~\ref{thm:nested}. Within that block $b_3=3\E[e^2]=3\sigma^2\neq0$. The residual of $\xi^3$ after projection onto $\operatorname{span}\{\xi\}$ is $\xi^3-(\mu_4/\sigma^2)\xi$, whose $b$-component is $3\sigma^2-\mu_4/\sigma^2=-\sigma^2\gamma_4$; it is non-zero precisely when $\E[\xi^4]\neq3(\E[\xi^2])^2$, i.e.\ $\gamma_4\neq0$. Evaluating $b^{\!\top}F_3^{-1}b$ on the block gives \eqref{eq:g3}; since $\gamma_4^2>0$ and the body factor $6+9\gamma_4+\gamma_6$ is positive, $g_3<1$.
\end{proof}

Symmetry shifts the informative content from the third cumulant, which the second-order span reads, to the fourth, which only a cubic basis function reads.

\subsection{The general degree-three coefficient}
\label{subsec:asym-reserve}

That degree-three PMM reduces the slope variance under asymmetric errors was established in \citet{zabolotnii2018asym}; the contribution here is the explicit closed form in the four standardized cumulants, its symmetric reduction, and their machine-checked verification. Standardizing to $\sigma=1$, so that $m_2=1$, $m_3=\gamma_3$, $m_4=\gamma_4+3$, $m_5=\gamma_5+10\gamma_3$, $m_6=\gamma_6+15\gamma_4+10\gamma_3^2+15$, the body and sensitivity vector are
\begin{equation}
F_3=\begin{pmatrix}
1&\gamma_3&\gamma_4+3\\
\gamma_3&\gamma_4+2&\gamma_5+9\gamma_3\\
\gamma_4+3&\gamma_5+9\gamma_3&\gamma_6+15\gamma_4+9\gamma_3^2+15
\end{pmatrix},\qquad b=(1,0,3)^{\!\top},
\label{eq:F3-matrix}
\end{equation}
and $g_3=\det F_3/\bigl(b^{\!\top}\operatorname{adj}(F_3)b\bigr)=N/D$ with
\begin{align}
N&=-9\gamma_3^4+12\gamma_3^2\gamma_4-\gamma_3^2\gamma_6-24\gamma_3^2+2\gamma_3\gamma_4\gamma_5-12\gamma_3\gamma_5\nonumber\\
&\qquad{}-\gamma_4^3+7\gamma_4^2+\gamma_4\gamma_6+24\gamma_4-\gamma_5^2+2\gamma_6+12,\label{eq:g3-num}\\
D&=9\gamma_3^2\gamma_4-18\gamma_3^2-12\gamma_3\gamma_5+9\gamma_4^2+\gamma_4\gamma_6-\gamma_5^2+24\gamma_4+2\gamma_6+12.\label{eq:g3-den}
\end{align}

\begin{theorem}[General degree-three coefficient]
\label{thm:asym}
Under A2--A3 with $S=3$, $g_3=N/D$ with $N,D$ as in \eqref{eq:g3-num}--\eqref{eq:g3-den}; $g_3$ depends on all of $(\gamma_3,\gamma_4,\gamma_5,\gamma_6)$, whereas $g_2$ depends on $(\gamma_3,\gamma_4)$ only, and $\gamma_3=\gamma_5=0$ recovers \eqref{eq:g3}.
\end{theorem}

\begin{proof}
$N=\det F_3$ and $D=b^{\!\top}\operatorname{adj}(F_3)b$ are polynomial identities in the entries of \eqref{eq:F3-matrix}, machine-checked (\S\ref{subsec:lean}); $D-N=9\gamma_3^4-3\gamma_3^2\gamma_4+\gamma_3^2\gamma_6+6\gamma_3^2-2\gamma_3\gamma_4\gamma_5+\gamma_4^3+2\gamma_4^2$ exhibits the four-cumulant dependence. For $\gamma_3=\gamma_5=0$, $D=(\gamma_4+2)(6+9\gamma_4+\gamma_6)$ and $D-N=\gamma_4^2(\gamma_4+2)$, and the common factor cancels. Positive definiteness of $F_3$ under A3 gives $N>0$ and $D>0$; $g_3\le g_2$ is Theorem~\ref{thm:nested}.
\end{proof}

Table~\ref{tab:asym-reserve} evaluates the closed form on four right-skewed laws: the additional asymptotic reserve $g_2/g_3-1$ over the second-order span ranges from $30\%$ on $\mathrm{Gamma}(2,1)$ to $50\%$ on the exponential. These are population constants for these four laws, not a distribution-free magnitude; an oracle Monte Carlo with population coefficients $h=F_3^{-1}b$ ($M=10\,000$, Supplement~S5, Table~S4) reproduces them to within two bootstrap standard errors.

\begin{table}[t]
\centering
\caption{\emph{Asymptotic} degree-three reserve over the second-order span on four right-skewed laws, from \eqref{eq:g3-num}--\eqref{eq:g3-den}. $1/g_2$ and $1/g_3$ are the asymptotic relative efficiencies over OLS of the degree-two (SLS/PMM$_2$) and degree-three estimators; the parenthesized figure is $g_2/g_3-1$. Population constants; no sampling involved.}
\label{tab:asym-reserve}
\begin{tabular}{lccccc}
\toprule
error law & $\gamma_3$ & $\gamma_4$ & $\gamma_5$ & $\gamma_6$ & $1/g_2\;\to\;1/g_3$ (reserve)\\
\midrule
$\chi^2(2)=\mathrm{Exp}$ & $2.00$ & $6.0$ & $24.0$ & $120.0$ & $2.00\;\to\;3.00$\ \ ($+50\%$)\\
$\mathrm{Gamma}(1.2,1)$ & $1.83$ & $5.0$ & $18.3$ & $83.3$ & $1.91\;\to\;2.76$\ \ ($+45\%$)\\
$\chi^2(3)$ & $1.63$ & $4.0$ & $13.1$ & $53.3$ & $1.80\;\to\;2.49$\ \ ($+38\%$)\\
$\mathrm{Gamma}(2,1)$ & $1.41$ & $3.0$ & $8.49$ & $30.0$ & $1.67\;\to\;2.17$\ \ ($+30\%$)\\
\bottomrule
\end{tabular}
\end{table}

\subsection{What the second-order span cannot read}
\label{subsec:span}

\begin{proposition}[Symmetric-kurtosis blindness of the second-order span]
\label{prop:blind}
Under A1--A4 with $S=2$ and a symmetric error law, the optimal slope combination is $F_2^{-1}b=(1/\sigma^2,0)^{\!\top}$: the slope estimating equation of SLS and of degree-two PMM places zero weight on the squared residual, reduces to the OLS normal equation, and carries no information from $\gamma_4$ or any even cumulant beyond $\sigma^2$.
\end{proposition}

The proof (Supplement~S2) notes that under symmetry $F_2$ is diagonal while the conditional cross-moment $\E[\rho_{i1}\rho_{i2}\mid x_i]=2\mu_i\sigma^2$ is not zero: the mechanism is not a block-diagonal weight but the annihilation of $\rho_2$ in the slope direction, $W_i^\star\partial_\beta\rho_i\propto T_i^{-\top}F_2^{-1}b$. The symmetric-kurtosis direction is $\xi^3-\operatorname{proj}_{\{\xi\}}\xi^3\in\mathcal H_3\setminus\mathcal H_2$, $L^2$-orthogonal to everything the span $\{e,e^2-\sigma^2\}$ can form.

Two qualifications keep the statement to its proper size. First, the reserve belongs to the span, not to the algorithm: by the same projection argument, any estimator that is efficient on the degree-three span attains $1/g_3$, in particular a two-step GMM on the six conditions $\{e,xe,e^2-\sigma^2,x(e^2-\sigma^2),e^3-\mu_3,x(e^3-\mu_3)\}$. What distinguishes the PMM route is the normal system: three coefficients with a fixed sensitivity vector $b$ in place of a $6\times6$ weight matrix estimated from the same sample, a difference that is invisible asymptotically and visible at $n\le200$ (Section~\ref{subsec:mc-reserve}). Second, the direction is unreadable by the second-order span, not by semiparametric estimation: in Case~2 of \citet{kim2019semiparametric} the efficient score involves the error log-density and its bound lies below $1/g_3$ whenever the law is non-Gaussian, and under a symmetry assumption (their Case~3) the efficient estimator is different again. Degree-three PMM reads the part of that gap that a cubic moment function can reach without a density.

\subsection{The feasible degree-$S$ estimator and the choice of degree}
\label{subsec:feasible}

The closed forms above use the population coefficients $h=F_S^{-1}b$. The estimator actually implemented replaces them by sample moments. Let $\hat\beta^{(0)}$ be OLS, $\hat e_i=y_i-x_i^{\!\top}\hat\beta^{(0)}$, $\hat\alpha_k=N^{-1}\sum_i\hat e_i^k$ for $k\le2S$, $\widehat F_S$ and $\hat b$ the body and sensitivity vector formed from $\hat\alpha$, and $\hat h=\widehat F_S^{-1}\hat b$. The feasible degree-$S$ estimator $\hat\beta_S$ solves
\begin{equation}
\sum_{i=1}^N (x_i-\bar x)\;\hat h^{\!\top}\bigl(m_S(y_i-x_i^{\!\top}\beta)-\hat\alpha_{1:S}\bigr)=0,\qquad m_S(e)=(e,e^2,\dots,e^S)^{\!\top},
\label{eq:feasible}
\end{equation}
for the slope, with the intercept equation $\sum_i\hat h^{\!\top}(m_S(\cdot)-\hat\alpha_{1:S})=0$; one Newton step from $\hat\beta^{(0)}$ or iteration to convergence, with the moments re-estimated at the updated $\hat\beta$, give the same first-order result.

\begin{proposition}[Feasible plug-in estimator]
\label{prop:feasible}
Under A1--A4, $\hat\beta_S$ is consistent and
\[
\sqrt N(\hat\beta_{S,1}-\beta_1)\;\xrightarrow{d}\;\mathcal N(0,\;c_2g_S),
\]
the same first-order law as the estimator with the population coefficients $h=F_S^{-1}b$. Finite moments of order $2S$ suffice: sixth moments for $S=3$. Finite moments of order $4S$ are needed only for the $\sqrt N$-asymptotic normality of $\hat h$ and $\hat g_S$ themselves (a confidence interval for the estimated reserve), not for the law of $\hat\beta_S$.
\end{proposition}

The proof is in Supplement~S3; three facts drive it. (i) The centered design in \eqref{eq:feasible} makes $\sum_i(x_i-\bar x)=0$, so the estimated means $\hat\alpha_{1:S}$ drop out of the slope equation exactly, and the plug-in enters only through $\hat h$. (ii) The derivative of the estimating function with respect to $h$ at the truth, $N^{-1}\sum_i(x_i-\bar x)m_S(e_i)^{\!\top}$, has mean zero under A1 and is $O_p(N^{-1/2})$, so an error $\hat h-h=o_p(1)$ perturbs the equation by $o_p(N^{-1/2})$ and leaves the first-order law unchanged; this is the standard argument for estimated nuisance quantities with zero expected derivative \citep[Section~6]{newey1994large}. (iii) $\hat h\to_ph$ needs only $\hat\alpha_k\to_p\alpha_k$ for $k\le2S$, which is a law of large numbers under A2 together with the OLS rate for $\hat\beta^{(0)}$ and A4, and continuity of $h(\alpha)=F_S(\alpha)^{-1}b(\alpha)$ on $\{\Delta_S>0\}$; uniqueness of the selected direction is A3. The finite-sample cost that \citet{kim2012efficiency} anticipate for estimated higher moments is real, but it is a second-order effect, and its sign turns out to depend on the law (Section~\ref{subsec:mc-feasible}).

\emph{Choosing between $S=2$ and $S=3$.} A practitioner needs a rule. The one we use is a pretest on the estimated reserve: compute $\hat d=1/\hat g_3-1/\hat g_2$ from the OLS residuals and its bootstrap standard deviation $\hat s$ ($B=200$ residual resamples), and take $S=3$ if $\hat d>2\hat s$, else $S=2$. The rule is conservative by design: it pays the degree-three estimation cost only when the reserve is resolved by the data, it selects $S=2$ under Gaussian errors at every sample size, and its realized efficiency is never below that of degree two in our experiments (Section~\ref{subsec:mc-feasible}). We also tested direction-specific stabilization of $\hat h$ along the $\hat m_6$ coordinate, which the submitted version conjectured might help; it does not, and the results are reported rather than omitted.

\section{Monte Carlo evidence and data examples}
\label{sec:montecarlo}

\subsection{Design and estimators}
\label{subsec:mc-setup}

Every regression cell uses
\begin{equation}
y_i=\beta_0+\beta_1x_i+e_i,\qquad \beta_0=2,\ \beta_1=1.5,\ x_i\sim\mathrm U(0,5),
\label{eq:dgp}
\end{equation}
with the slope $\beta_1$ as estimand and common random numbers within a cell. The error laws are centered so that $\E[e]=0$ and are otherwise left at their natural scale, since every estimator is scale-equivariant: $\chi^2(3)-3$; $\mathrm{Gamma}(k,1)-k$ for $k\in\{1,2,4,8\}$; a standardized $\mathrm{lognormal}(0,0.5)$; $\mathrm U(-\sqrt3,\sqrt3)$; $(\mathrm{Beta}(2,2)-\tfrac12)\sqrt{20}$; Student $t_\nu$ scaled to unit variance for $\nu\in\{10,7,5\}$ (the sixth moment is infinite for $\nu\le6$, so $t_5$ is a deliberate violation of A2); and $\mathcal N(0,1)$. The asymptotic constants $1/g_2$, $1/g_3$ quoted next to each cell are computed from the analytic cumulants where available and otherwise from $4\times10^6$ draws; they agree with Table~\ref{tab:asym-reserve} to the digits shown.

Seven estimators of $\beta_1$ are compared. OLS is the baseline. PMM$_2$ and PMM$_3$ are the feasible plug-in estimators of \S\ref{subsec:feasible} with the full degree-$S$ body, solved by Newton's method on \eqref{eq:feasible}; the public \textsf{EstemPMM} package \citep{estempmmcran} is used as a cross-check for degree two, and its degree-three routine, which assumes a symmetric error and uses the even-moment body, is reported in Supplement~S5 on every law to document that it is valid on the symmetric laws only. SLS is the feasible optimally weighted Wang--Leblanc estimator with the regression-form second residual and $\widehat W_i=\widehat T_i^{-\top}\widehat F_2^{-1}\widehat T_i^{-1}$ from OLS residual moments; PMM$_2$ and SLS thus share $\widehat F_2$, and any finite-sample gap between them is a solver effect. GMM$_2$ and GMM$_3$ are two-step efficient GMM estimators \citep{hansen1982large} on the four degree-two conditions and on the six degree-three conditions of \S\ref{subsec:span}, respectively. All efficiencies below are Monte Carlo $\RE=\mathrm{MSE}_{\OLS}/\mathrm{MSE}$ of \emph{feasible} estimators unless a caption says asymptotic or oracle; the standard errors in parentheses are paired bootstrap standard errors over replications. Convergence was $100\%$ in every cell and $|\mathrm{bias}|<0.02$ throughout. Scripts and result files are in the code supplement; the complete tables are in Supplement~S5.

\subsection{Equivalence of PMM$_2$ and SLS under asymmetric errors}
\label{subsec:mc-h1}

Table~\ref{tab:e1-asym} reports the asymmetric regime ($M=2000$). PMM$_2$ and SLS are indistinguishable at every $n$, GMM$_2$ joins them from $n=100$ on (at $n=50$ its estimated $4\times4$ weight costs it between $0.1$ and $0.3$ in $\RE$), and from $n=500$ on all three sit within one standard error of the asymptotic value $1/g_2$. The equivalence is stronger than equal variances: in the original pilot study (Supplement~S5, Table~S2) the per-replication correlation between the PMM$_2$ and SLS slopes rises from $0.996$ at $n=100$ to $0.9997$ at $n=1000$ and the paired difference contracts toward zero, which is the signature of a shared influence function (Figure~\ref{fig:h1}). The finite-sample values of the second-order estimators at $n\le200$ lie above $1/g_2$; this is the same second-order effect as for degree three, discussed in \S\ref{subsec:mc-feasible}.

\begin{table}[t]
\centering
\caption{\emph{Feasible} Monte Carlo relative efficiency over OLS under asymmetric errors, design \eqref{eq:dgp}, $M=2000$ replications, paired bootstrap standard errors in parentheses. All estimators use moments estimated from OLS residuals. The asymptotic constants $1/g_2$ and $1/g_3$ are given in the row labels. The last column is the per-replication correlation between the PMM$_3$ and GMM$_3$ slopes.}
\label{tab:e1-asym}
\scriptsize
\setlength{\tabcolsep}{3pt}
\begin{tabular}{@{}llcccccc@{}}
\toprule
Error & $n$ & PMM$_2$ & SLS & GMM$_2$ & PMM$_3$ & GMM$_3$ & cor$_{3}$\\
\midrule
\multicolumn{8}{@{}l}{$\chi^2(3)$: $1/g_2=1.80$, $1/g_3=2.50$}\\
 & 50 & $2.12\,(0.07)$ & $2.05\,(0.06)$ & $1.82\,(0.06)$ & $3.05\,(0.12)$ & $1.52\,(0.05)$ & 0.686 \\
 & 100 & $1.97\,(0.06)$ & $1.97\,(0.06)$ & $1.91\,(0.05)$ & $2.96\,(0.11)$ & $2.12\,(0.06)$ & 0.833 \\
 & 200 & $1.90\,(0.06)$ & $1.89\,(0.06)$ & $1.89\,(0.06)$ & $3.01\,(0.12)$ & $2.66\,(0.09)$ & 0.929 \\
 & 500 & $1.80\,(0.05)$ & $1.79\,(0.05)$ & $1.80\,(0.05)$ & $2.71\,(0.10)$ & $2.62\,(0.08)$ & 0.978 \\
 & 1000 & $1.83\,(0.05)$ & $1.83\,(0.05)$ & $1.85\,(0.05)$ & $2.74\,(0.10)$ & $2.72\,(0.09)$ & 0.987 \\
\midrule
\multicolumn{8}{@{}l}{$\mathrm{Gamma}(2,1)$: $1/g_2=1.67$, $1/g_3=2.17$}\\
 & 50 & $1.72\,(0.05)$ & $1.73\,(0.05)$ & $1.58\,(0.04)$ & $2.27\,(0.08)$ & $1.34\,(0.04)$ & 0.748 \\
 & 100 & $1.76\,(0.05)$ & $1.75\,(0.05)$ & $1.72\,(0.05)$ & $2.32\,(0.09)$ & $1.82\,(0.05)$ & 0.878 \\
 & 200 & $1.74\,(0.05)$ & $1.72\,(0.05)$ & $1.75\,(0.05)$ & $2.46\,(0.08)$ & $2.22\,(0.07)$ & 0.942 \\
 & 500 & $1.81\,(0.05)$ & $1.80\,(0.05)$ & $1.82\,(0.05)$ & $2.40\,(0.09)$ & $2.36\,(0.08)$ & 0.982 \\
 & 1000 & $1.68\,(0.05)$ & $1.68\,(0.05)$ & $1.69\,(0.05)$ & $2.22\,(0.07)$ & $2.22\,(0.07)$ & 0.992 \\
\bottomrule
\end{tabular}
\end{table}

\begin{figure}[t]
  \centering
  \includegraphics[width=\textwidth]{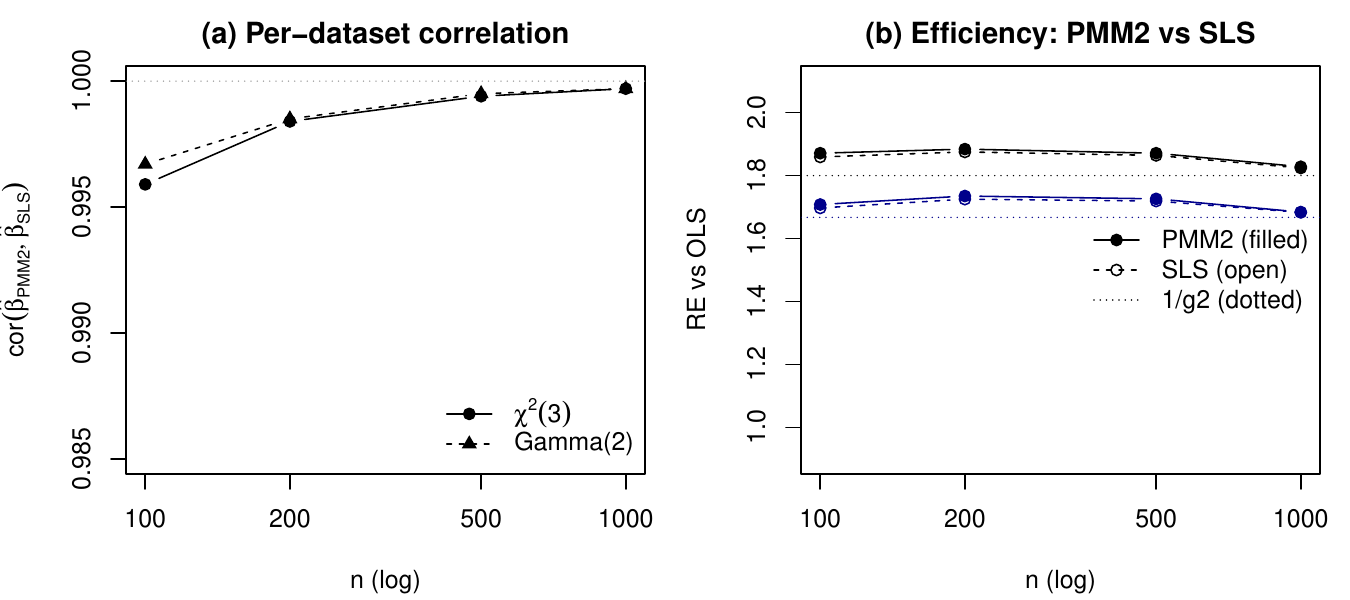}
  \caption{PMM$_2$ and SLS become indistinguishable as $n$ grows under asymmetric errors (pilot study, $M=500$, Supplement~S5). (a) Per-replication correlation of the two slope estimates. (b) Realized feasible efficiencies of PMM$_2$ (filled) and SLS (open); dotted lines mark the asymptotic value $1/g_2$.}
  \label{fig:h1}
\end{figure}

\subsection{The degree-three reserve and the degree-three GMM}
\label{subsec:mc-reserve}

The right-hand block of Table~\ref{tab:e1-asym} and Table~\ref{tab:e1-sym} address the reserve. On the skewed laws PMM$_3$ exceeds every degree-two estimator at every $n$, by the margin Table~\ref{tab:asym-reserve} predicts once $n\ge500$. GMM$_3$ confirms that the reserve is a property of the span: at $n\ge500$ it coincides with PMM$_3$ (correlation $0.98$--$0.99$), whereas at $n\le200$ its estimated $6\times6$ weight matrix costs it most of the gain ($1.5$ against $3.1$ at $n=50$ on $\chi^2(3)$). On the uniform law, Theorem~\ref{thm:symm}'s regime, PMM$_2$, SLS and GMM$_2$ stay at the OLS line for every $n$ while PMM$_3$ climbs from $2.2$ at $n=50$ to $3.37$ (SE $0.12$) at $n=1000$, reaching the asymptotic $1/g_3=3.33$ from below, and GMM$_3$ follows one step behind. Under the Gaussian control (Supplement~S5, Table~S7) no estimator gains, and the degree-three estimators pay a small, shrinking finite-sample cost ($\RE$ of PMM$_3$ $0.82$ at $n=50$, $1.00$ at $n=1000$).

\begin{table}[t]
\centering
\caption{\emph{Feasible} Monte Carlo relative efficiency over OLS under symmetric platykurtic errors $\mathrm U(-\sqrt3,\sqrt3)$, design \eqref{eq:dgp}, $M=2000$, paired bootstrap standard errors in parentheses. Asymptotic constants $1/g_2=1$, $1/g_3=3.33$. The Gaussian control and the symmetric even-moment degree-three routine of \textsf{EstemPMM}, which agrees with PMM$_3$ on this law to within $0.1$, are in Supplement~S5, Table~S7.}
\label{tab:e1-sym}
\scriptsize
\setlength{\tabcolsep}{3pt}
\begin{tabular}{@{}lcccccc@{}}
\toprule
$n$ & PMM$_2$ & SLS & GMM$_2$ & PMM$_3$ & GMM$_3$ & cor$_{3}$\\
\midrule
50 & $0.93\,(0.01)$ & $0.89\,(0.01)$ & $0.94\,(0.01)$ & $2.18\,(0.08)$ & $2.00\,(0.07)$ & 0.927 \\
100 & $0.96\,(0.01)$ & $0.95\,(0.01)$ & $0.98\,(0.01)$ & $2.78\,(0.11)$ & $2.50\,(0.09)$ & 0.972 \\
200 & $0.99\,(0.00)$ & $0.99\,(0.01)$ & $0.99\,(0.01)$ & $2.99\,(0.11)$ & $2.83\,(0.10)$ & 0.988 \\
500 & $0.99\,(0.00)$ & $0.99\,(0.00)$ & $1.00\,(0.00)$ & $3.18\,(0.12)$ & $3.15\,(0.12)$ & 0.995 \\
1000 & $0.99\,(0.00)$ & $0.99\,(0.00)$ & $1.00\,(0.00)$ & $3.37\,(0.12)$ & $3.33\,(0.12)$ & 0.998 \\
\bottomrule
\end{tabular}
\end{table}

\subsection{Feasible-sample sweep over skewness, kurtosis and tails}
\label{subsec:mc-feasible}

Table~\ref{tab:e2} follows the feasible PMM$_3$ over eleven error laws and six sample sizes against its asymptotic constant, with the Monte Carlo standard deviation of the estimated sixth standardized cumulant as the instability gauge (full grid in Supplement~S5). Three regularities emerge. First, on every law with a positive reserve the feasible degree-three estimator is more efficient than the feasible degree-two estimator at every $n\ge200$, at $n=100$ on all of them except $t_{10}$ and Gamma(8,1) --- where it trails by less than one Monte Carlo standard error ($1.16$ against $1.18$ and $0.98$ against $0.99$; Supplement~S5) --- and already at $n=50$ on the laws with $\gamma_3\ge1$: the reserve survives estimation of the moment structure. Second, the direction of the finite-sample departure from $1/g_3$ depends on the law. On the right-skewed laws the feasible $\RE$ lies \emph{above} the asymptotic constant and converges to it from above (Gamma(1): $3.9$ at $n=50$, $3.2$ at $n=2000$, against $3.00$); on the platykurtic laws it converges from below (uniform $2.2\to3.2$; Beta(2,2) $1.17\to1.52$ against $1.54$). The plug-in coefficient vector is computed from the same residuals it weights, and on skewed laws this dependence reduces the slope error at second order, an $O(1/n)$ effect that decays slowly; the asymptotic constant, not a small-sample table value, is the number to quote for a skewed law. Third, the sixth-moment instability is real, $\operatorname{sd}(\hat\gamma_6)$ being of the order of $\gamma_6$ on the Gamma and lognormal laws for $n\le500$, but it does not translate into a loss, because $\hat m_6$ enters only through the direction of $\hat h$. On $t_5$, where the sixth moment does not exist, the feasible PMM$_3$ still returns $\RE$ between $1.02$ and $1.18$ with no convergence failures, so the violation of A2 degrades rather than breaks the estimator; on $t_7$ and $t_{10}$ the reserve is asymptotically small ($1/g_3\le1.05$) and the realized values match it.

\begin{table}[t]
\centering
\caption{\emph{Feasible} PMM$_3$ across error laws and sample sizes, design \eqref{eq:dgp}, $M=2000$; Monte Carlo relative efficiency over OLS with paired bootstrap standard errors, next to the \emph{asymptotic} constants $1/g_2$ and $1/g_3$. $\RE_2$ is the feasible PMM$_2$ at $n=2000$. $\operatorname{sd}(\hat\gamma_6)$ is the Monte Carlo standard deviation of the sixth standardized cumulant estimated from OLS residuals at $n=200$. The $t_5$ row violates A2: its sixth moment is infinite, so $\gamma_6$ and with it the population constant $1/g_3$ do not exist, and that entry is left blank rather than filled with the value the finite formula returns on a truncated moment.}
\label{tab:e2}
\scriptsize
\setlength{\tabcolsep}{2.4pt}
\begin{tabular}{@{}lcccccccc@{}}
\toprule
Error law & $1/g_2$ & $\RE_2$ & $1/g_3$ & $\RE_3$, $n{=}50$ & $n{=}200$ & $n{=}1000$ & $n{=}2000$ & $\operatorname{sd}(\hat\gamma_6)$\\
\midrule
$\mathrm{Gamma}(1,1)$ & 2.00 & $1.95\,(0.06)$ & 3.00 & $3.90\,(0.17)$ & $3.53\,(0.14)$ & $3.63\,(0.14)$ & $3.16\,(0.12)$ & 194 \\
$\mathrm{Gamma}(2,1)$ & 1.67 & $1.61\,(0.05)$ & 2.17 & $2.27\,(0.08)$ & $2.46\,(0.08)$ & $2.22\,(0.07)$ & $2.20\,(0.07)$ & 66 \\
$\mathrm{Gamma}(4,1)$ & 1.40 & $1.41\,(0.04)$ & 1.60 & $1.44\,(0.05)$ & $1.66\,(0.04)$ & $1.61\,(0.04)$ & $1.64\,(0.04)$ & 30 \\
$\mathrm{Gamma}(8,1)$ & 1.22 & $1.18\,(0.02)$ & 1.29 & $1.10\,(0.03)$ & $1.27\,(0.03)$ & $1.32\,(0.03)$ & $1.24\,(0.03)$ & 11 \\
lognormal$(0,0.5)$ & 1.63 & $1.71\,(0.05)$ & 1.99 & $2.09\,(0.07)$ & $2.30\,(0.08)$ & $2.26\,(0.07)$ & $2.25\,(0.07)$ & 285 \\
$\mathrm U(-\sqrt3,\sqrt3)$ & 1.00 & $1.00\,(0.00)$ & 3.33 & $2.18\,(0.08)$ & $2.99\,(0.11)$ & $3.37\,(0.12)$ & $3.18\,(0.12)$ & 1 \\
$\mathrm{Beta}(2,2)$ & 1.00 & $1.00\,(0.00)$ & 1.54 & $1.17\,(0.03)$ & $1.43\,(0.04)$ & $1.51\,(0.04)$ & $1.52\,(0.04)$ & 1 \\
$t_{10}$ & 1.00 & $1.00\,(0.00)$ & 1.04 & $0.90\,(0.02)$ & $1.05\,(0.01)$ & $1.03\,(0.01)$ & $1.04\,(0.01)$ & 25 \\
$t_7$ & 1.00 & $1.01\,(0.00)$ & 1.05 & $0.93\,(0.02)$ & $1.10\,(0.01)$ & $1.08\,(0.01)$ & $1.08\,(0.01)$ & 261 \\
$t_5$ & 1.00 & $1.01\,(0.01)$ & --- & $1.02\,(0.03)$ & $1.18\,(0.02)$ & $1.15\,(0.02)$ & $1.11\,(0.02)$ & 252 \\
$\mathcal N(0,1)$ & 1.00 & $1.00\,(0.00)$ & 1.00 & $0.82\,(0.02)$ & $0.98\,(0.01)$ & $1.00\,(0.00)$ & $0.99\,(0.00)$ & 2 \\
\bottomrule
\end{tabular}
\end{table}

\subsection{Stabilization of the plug-in and the degree-selection rule}
\label{subsec:mc-select}

Table~\ref{tab:e3} compares, with one and the same Newton solver, the plain plug-in PMM$_3$, the oracle-coefficient PMM$_3$ with the population $h$, three stabilized plug-ins, and the pretest-selected degree ($M=1000$). The stabilized variants are a scalar Tikhonov ridge $\hat h=(\widehat F_3+\lambda I)^{-1}\hat b$ with $\lambda=0.1\operatorname{tr}(\widehat F_3)/3$; a direction-specific ridge that penalizes only the $(3,3)$ entry of $\widehat F_3$, i.e.\ the $\hat m_6$ coordinate, by $\lambda=0.5\,\widehat F_{33}$; and a James--Stein-type shrinkage of $\hat m_6$ toward its value under the closure $\kappa_6=0$, with weight $\hat v/\{\hat v+(\hat m_6-m_6^\ast)^2\}$ and $\hat v$ the estimated variance of $\hat m_6$. The plain plug-in sits at or above the oracle-coefficient estimator on the skewed laws at every $n$ (the second-order effect of \S\ref{subsec:mc-feasible}) and within one standard error of it on the uniform law; under Gaussian errors it pays for estimating a direction that does not exist ($0.82$ against $1.00$ at $n=50$, vanishing by $n=1000$). None of the three stabilized variants improves on the plain plug-in on the skewed laws: the scalar ridge and the $\kappa_6$-shrinkage are destructive there ($\RE$ between $0.1$ and $0.7$, because they move the weight off the informative low-order coordinates), and the $\hat m_6$-ridge loses $14$--$25\%$ of the plain estimator's efficiency at every $n$; on the uniform law the shrinkage is harmless and the ridges cost. We therefore recommend no regularization of the body. The pretest rule behaves as intended: it selects $S=3$ in fewer than $3\%$ of samples at $n=50$, in $86$--$100\%$ at $n\ge200$ on the laws with a reserve, and in at most $0.3\%$ on Gaussian errors at any $n$, so the selected estimator equals PMM$_2$ where the reserve is unresolved and PMM$_3$ where it is resolved, and is never below PMM$_2$ in efficiency.

\begin{table}[t]
\centering
\caption{Stabilized plug-ins and the pretest rule ($M=1000$, design \eqref{eq:dgp}): \emph{feasible} Monte Carlo relative efficiency over OLS with paired bootstrap standard errors. ``oracle $h$'' uses the \emph{population} coefficients $h=F_3^{-1}b$ with the same solver; ``$\hat m_6$-ridge'' penalizes only the sixth-moment coordinate of $\widehat F_3$; ``select'' is the pretest-selected degree $S\in\{2,3\}$ and $\hat p_3$ the fraction of samples in which it selected $S=3$. Scalar ridge and $\kappa_6$-shrinkage are in Supplement~S5, Table~S10.}
\label{tab:e3}
\scriptsize
\setlength{\tabcolsep}{3pt}
\begin{tabular}{@{}llcccccc@{}}
\toprule
Error & $n$ & PMM$_2$ & PMM$_3$ plain & PMM$_3$ oracle $h$ & PMM$_3$ $\hat m_6$-ridge & select & $\hat p_3$\\
\midrule
\multicolumn{8}{@{}l}{$\chi^2(3)$}\\
 & 50 & $1.95\,(0.10)$ & $2.67\,(0.16)$ & $2.35\,(0.14)$ & $2.21\,(0.12)$ & $1.96\,(0.10)$ & 0.003 \\
 & 100 & $2.03\,(0.08)$ & $3.09\,(0.16)$ & $2.61\,(0.12)$ & $2.32\,(0.10)$ & $2.15\,(0.09)$ & 0.200 \\
 & 200 & $1.83\,(0.08)$ & $2.79\,(0.15)$ & $2.31\,(0.11)$ & $2.09\,(0.09)$ & $2.69\,(0.14)$ & 0.937 \\
 & 500 & $1.94\,(0.08)$ & $2.69\,(0.13)$ & $2.49\,(0.12)$ & $2.20\,(0.09)$ & $2.69\,(0.13)$ & 1.000 \\
 & 1000 & $1.98\,(0.08)$ & $2.90\,(0.14)$ & $2.70\,(0.13)$ & $2.30\,(0.10)$ & $2.90\,(0.14)$ & 1.000 \\
\midrule
\multicolumn{8}{@{}l}{$\mathrm{Gamma}(2,1)$}\\
 & 50 & $1.81\,(0.07)$ & $2.52\,(0.13)$ & $2.26\,(0.11)$ & $2.03\,(0.09)$ & $1.81\,(0.07)$ & 0.000 \\
 & 100 & $1.74\,(0.07)$ & $2.42\,(0.12)$ & $2.18\,(0.10)$ & $1.95\,(0.08)$ & $1.74\,(0.07)$ & 0.077 \\
 & 200 & $1.81\,(0.08)$ & $2.33\,(0.12)$ & $2.17\,(0.11)$ & $2.00\,(0.09)$ & $2.12\,(0.11)$ & 0.857 \\
 & 500 & $1.68\,(0.07)$ & $2.38\,(0.12)$ & $2.18\,(0.10)$ & $1.91\,(0.08)$ & $2.38\,(0.12)$ & 0.997 \\
 & 1000 & $1.54\,(0.06)$ & $2.18\,(0.10)$ & $2.04\,(0.10)$ & $1.75\,(0.07)$ & $2.18\,(0.10)$ & 1.000 \\
\midrule
\multicolumn{8}{@{}l}{$\mathrm U(-\sqrt3,\sqrt3)$}\\
 & 50 & $0.94\,(0.01)$ & $2.20\,(0.12)$ & $2.10\,(0.14)$ & $1.55\,(0.04)$ & $0.94\,(0.01)$ & 0.027 \\
 & 100 & $0.98\,(0.01)$ & $3.07\,(0.17)$ & $2.98\,(0.18)$ & $1.93\,(0.05)$ & $1.39\,(0.05)$ & 0.667 \\
 & 200 & $0.99\,(0.01)$ & $2.89\,(0.16)$ & $2.86\,(0.16)$ & $1.98\,(0.05)$ & $2.89\,(0.16)$ & 1.000 \\
 & 500 & $0.99\,(0.00)$ & $3.32\,(0.17)$ & $3.38\,(0.18)$ & $2.08\,(0.04)$ & $3.32\,(0.17)$ & 1.000 \\
 & 1000 & $0.99\,(0.00)$ & $3.60\,(0.19)$ & $3.61\,(0.19)$ & $2.15\,(0.05)$ & $3.60\,(0.19)$ & 1.000 \\
\midrule
\multicolumn{8}{@{}l}{$\mathcal N(0,1)$}\\
 & 50 & $0.95\,(0.01)$ & $0.82\,(0.02)$ & $1.00\,(0.00)$ & $0.93\,(0.01)$ & $0.95\,(0.01)$ & 0.000 \\
 & 100 & $0.95\,(0.01)$ & $0.90\,(0.02)$ & $1.00\,(0.00)$ & $0.94\,(0.01)$ & $0.95\,(0.01)$ & 0.000 \\
 & 200 & $0.98\,(0.01)$ & $0.95\,(0.01)$ & $1.00\,(0.00)$ & $0.97\,(0.01)$ & $0.98\,(0.01)$ & 0.000 \\
 & 500 & $0.99\,(0.01)$ & $0.98\,(0.01)$ & $1.00\,(0.00)$ & $0.99\,(0.01)$ & $0.99\,(0.01)$ & 0.000 \\
 & 1000 & $1.00\,(0.00)$ & $0.99\,(0.01)$ & $1.00\,(0.00)$ & $1.00\,(0.01)$ & $1.00\,(0.00)$ & 0.003 \\
\bottomrule
\end{tabular}
\end{table}

\subsection{Beyond A1: heteroskedasticity}
\label{subsec:mc-hetero}

Proposition~\ref{prop:hetero} is checked with two designs ($M=2000$; full table in Supplement~S5, Table~S11). Design~A is symmetric with an even variance function, $x\sim\mathrm U(-2.5,2.5)$, $e\mid x=\sigma(x)z$, $\sigma^2(x)=1+\tfrac12x^2$, $z\sim\mathcal N(0,1)$, so $\operatorname{Cov}(x,\sigma^2(x))=0$; Design~B is asymmetric with a monotone variance, $x\sim\mathrm U(0,5)$, $\sigma^2(x)=\exp\{0.4(x-2.5)\}$, $z$ a standardized $\chi^2(3)$, so $\operatorname{Cov}(x,\sigma^2(x))\neq0$ and $\kappa_3\neq0$. In Design~A every estimator is unbiased; OLS and unconditional PMM$_2$ share $\RE=1$, while conditional SLS (variance model $\exp(\tau_0+\tau_1x+\tau_2x^2)$, then weighted least squares) and oracle weighted least squares reach $\RE\approx1.09$ at $n=2000$, the Aitken gain the unconditional body cannot see. In Design~B OLS, conditional SLS and oracle WLS remain unbiased, the latter two with $\RE\approx1.3$, whereas unconditional PMM$_2$ carries a slope bias of about $-0.10$ that does not shrink with $n$ ($-0.109,-0.100,-0.094$ at $n=100,500,2000$) and its $\RE$ --- which is the MSE ratio of \S\ref{subsec:mc-setup} throughout, so a bias that does not shrink enters it --- collapses ($0.42,0.11,0.035$): the inconsistency of Proposition~\ref{prop:hetero}(ii). Outside A1, conditional SLS is the more general object.

\subsection{Two public data sets and a drilling-calibrated error law}
\label{subsec:real}

Theorems~\ref{thm:sls-pmm2}--\ref{thm:asym} predict the sign of the gains on real residual distributions; the population constants cannot be checked on a single data set, so we report, for each set, the residual cumulants, the estimates with paired case-bootstrap standard errors ($B=2000$), and a residual-calibrated Monte Carlo in which errors are resampled from the centered OLS residuals and the design from the observed regressor ($M=2000$), at the data's own $n$ and at $n=1000$. Details and diagnostics are in Supplement~S6.

\emph{cars} \citep{ezekiel1930methods} regresses stopping distance on speed ($n=50$); the OLS residuals are right-skewed ($\hat\gamma_3=0.89$, $\hat\gamma_4=0.89$), and the Breusch--Pagan test \citep{breusch1979simple} does not reject homoskedasticity ($p=0.07$). \emph{faithful} \citep{azzalini1990look} regresses eruption duration on waiting time ($n=272$); the residuals are nearly symmetric and platykurtic ($\hat\gamma_3=-0.14$, $\hat\gamma_4=-0.45$; $p=0.21$). Table~\ref{tab:real} reports the results. On \emph{cars} the estimated reserve is $1/\hat g_2=1.37$ and $1/\hat g_3=1.43$; the calibrated Monte Carlo gives $\RE$ of $1.30$ for PMM$_2$, $1.27$ for SLS and $1.21$ for PMM$_3$ at $n=50$, and $1.41$, $1.41$, $1.46$ at $n=1000$, so at the data's own size the degree-two gain is realized and the degree-three surcharge is not. On \emph{faithful} the degree-two channel is empty by symmetry ($1/\hat g_2=1.01$) and the platykurtic channel small ($1/\hat g_3=1.10$); the calibrated $\RE$ of PMM$_3$ is $1.04$ at $n=272$ and $1.09$ at $n=1000$, while PMM$_2$ and SLS stay at $0.99$--$1.01$, which is Theorem~\ref{thm:symm} at the scale of the data.

\begin{table}[t]
\centering
\caption{Two public regression data sets. Top: OLS residual skewness and excess kurtosis, Breusch--Pagan $p$-value, and the \emph{estimated} asymptotic efficiencies $1/\hat g_2$, $1/\hat g_3$ from residual moments. Middle: slope estimates with paired case-bootstrap standard errors ($B=2000$). Bottom: residual-calibrated Monte Carlo $\RE$ over OLS ($M=2000$) at the data's $n$ and at $n=1000$; paired bootstrap standard errors of these $\RE$ values are at most $0.04$ (Supplement~S6).}
\label{tab:real}
\footnotesize
\begin{tabular}{@{}lcc@{}}
\toprule
& \emph{cars} ($n=50$) & \emph{faithful} ($n=272$)\\
\midrule
residual $\hat\gamma_3$, $\hat\gamma_4$ & $0.89$, $0.89$ & $-0.14$, $-0.45$\\
Breusch--Pagan $p$ & $0.07$ & $0.21$\\
estimated $1/\hat g_2$, $1/\hat g_3$ & $1.37$, $1.43$ & $1.01$, $1.10$\\
\midrule
slope, OLS (SE) & $3.93\,(0.41)$ & $0.0756\,(0.0019)$\\
slope, PMM$_2$ (SE) & $3.44\,(0.48)$ & $0.0759\,(0.0020)$\\
slope, PMM$_3$ (SE) & $3.23\,(0.50)$ & $0.0760\,(0.0020)$\\
slope, SLS (SE) & $3.60\,(0.40)$ & $0.0756\,(0.0019)$\\
\midrule
calibrated $\RE$, PMM$_2$: $n_{\text{data}}$ / $1000$ & $1.30$ / $1.41$ & $0.99$ / $1.01$\\
calibrated $\RE$, SLS & $1.27$ / $1.41$ & $0.99$ / $1.00$\\
calibrated $\RE$, PMM$_3$ & $1.21$ / $1.46$ & $1.04$ / $1.09$\\
\bottomrule
\end{tabular}
\end{table}

\emph{A drilling-calibrated error law.} Field regressions of rate of penetration on weight on bit are a standard engineering use of least squares with skewed, heavy-tailed residuals. We use the processed 1-ft Pason log of Utah FORGE well 58-32, published under CC~BY~4.0 \citep{podgorney2018forge}. Raw residuals of such a regression violate A1: within a depth window their lag-one autocorrelation is $0.35$--$0.80$, so the data are \emph{not} used as a regression sample. Instead the empirical error law is extracted: in each of fourteen 500-ft windows on the dominant rotary-speed plateau, OLS of rate of penetration on weight on bit is followed by AR(1) prewhitening of the residuals, and the standardized innovations are pooled ($n=3266$; window table in Supplement~S6). The pooled law is right-skewed and heavy-tailed ($\hat\gamma_3=0.83$, $\hat\gamma_4=8.2$, $\hat\gamma_6=375$), with estimated asymptotic efficiencies $1/\hat g_2=1.07$ and $1/\hat g_3=1.17$; after a three-inter-quartile-range Tukey fence that removes $0.9\%$ of the innovations the law is mildly skewed ($\hat\gamma_3=0.18$, $\hat\gamma_4=1.4$) with $1/\hat g_2=1.01$ and $1/\hat g_3=1.14$. Errors resampled from either law drive the design \eqref{eq:dgp} (Table~\ref{tab:forge}). On the pooled law PMM$_3$ reaches $\RE$ of $1.29$, $1.34$ and $1.25$ at $n=100$, $250$ and $1000$, PMM$_2$ and SLS $1.08$--$1.13$, and GMM$_3$ $1.06$, $1.29$ and $1.26$, the same small-sample lag as in Table~\ref{tab:e1-asym}; on the fenced law the degree-two estimators have nothing to read ($\RE\approx1.00$) while PMM$_3$ keeps $1.09$--$1.16$. On real drilling innovations, then, the second-order span is nearly empty and the cubic direction holds a $10$--$16\%$ reserve that a feasible estimator realizes from $n=100$.

\begin{table}[t]
\centering
\caption{Residual-calibrated Monte Carlo with the FORGE 58-32 innovation law ($M=2000$, design \eqref{eq:dgp}): \emph{feasible} $\RE$ over OLS with paired bootstrap standard errors. ``pooled'' is the law of all $3266$ prewhitened innovations; ``fenced'' removes $0.9\%$ beyond a 3-IQR Tukey fence.}
\label{tab:forge}
\scriptsize
\setlength{\tabcolsep}{3pt}
\begin{tabular}{@{}llccccc@{}}
\toprule
law & $n$ & PMM$_2$ & SLS & PMM$_3$ & GMM$_3$ & $1/\hat g_2$, $1/\hat g_3$\\
\midrule
pooled ($n_{\text{law}}=3266$)
 & 100 & $1.11\,(0.02)$ & $1.12\,(0.02)$ & $1.29\,(0.03)$ & $1.06\,(0.04)$ & $1.07,\ 1.17$ \\
 & 250 & $1.13\,(0.02)$ & $1.13\,(0.02)$ & $1.34\,(0.03)$ & $1.29\,(0.03)$ & $1.07,\ 1.17$ \\
 & 1000 & $1.08\,(0.01)$ & $1.08\,(0.01)$ & $1.25\,(0.02)$ & $1.26\,(0.03)$ & $1.07,\ 1.17$ \\
\midrule
fenced ($n_{\text{law}}=3237$)
 & 100 & $1.01\,(0.01)$ & $1.00\,(0.01)$ & $1.09\,(0.02)$ & $1.00\,(0.02)$ & $1.01,\ 1.14$ \\
 & 250 & $1.00\,(0.01)$ & $1.00\,(0.01)$ & $1.11\,(0.02)$ & $1.08\,(0.02)$ & $1.01,\ 1.14$ \\
 & 1000 & $1.01\,(0.01)$ & $1.01\,(0.01)$ & $1.16\,(0.02)$ & $1.16\,(0.02)$ & $1.01,\ 1.14$ \\
\bottomrule
\end{tabular}
\end{table}

\section{Discussion}
\label{sec:discussion}

\subsection{One family read at one degree}

Theorems~\ref{thm:sls-pmm2} and~\ref{thm:nested} place SLS and PMM on a single axis indexed by the polynomial degree: the estimating family is generated by $\{\xi,\dots,\xi^S\}$, the optimal coefficients solve $F_Sh=b$, and enlarging the basis cannot increase the variance ratio. Optimally weighted SLS is the $S=2$ instance, with $F_2$ as its weighting operator and one influence function shared with degree-two PMM; it dominates a same-span GMM with a sub-optimal weight and equals the optimal-weight GMM at convergence. What the simulations establish is the equivalence itself, the fingerprints of one estimating equation, not of two methods that perform alike.

\subsection{The reserve as a property of the span}

The headroom above $S=2$ within the polynomial class is the cubic moment direction that the second-order span does not contain; under symmetry it is the only informative direction, and the second-order estimators reduce to OLS for the slope in the setting of Theorem~\ref{thm:sls-pmm2}. It is a property of the enlarged span, reached asymptotically by any efficient estimator on it, GMM$_3$ included; the PMM normal system is a low-dimensional route to it that pays at small $n$. It is not a semiparametric bound: in Case~2 of \citet{kim2019semiparametric} the efficient score uses the error log-density and its bound is lower still. The claim of this paper is bounded accordingly: within the polynomial moment class, degree three is where the second-order span stops and the next informative direction begins, and Tables~\ref{tab:e1-asym}--\ref{tab:forge} show that a feasible estimator collects that direction at sample sizes of practical interest.

\subsection{Orthogonal axes of extension}
\label{subsec:axes}

The degree axis is one of several along which the normal-system core can be moved. The basis need not consist of integer powers, and one may change the basis family, add regularization, or make the body conditional on covariates as \S\ref{subsec:hetero} suggests. These are orthogonal extensions: they vary the basis or the conditioning while holding the degree machinery fixed, whereas this paper varies the degree while holding the integer-power basis fixed. The theorems are statements about the degree axis, not about a software implementation or a future basis family.

\subsection{Limitations: five boundaries}
\label{subsec:limitations}

Five boundaries delimit the claims. First, the closed forms are degree-specific: $g_2$ and $g_3$ are compact, but for $S\ge4$ the ratio is the solution of a larger normal system, and the ordering for higher degrees rests on Theorem~\ref{thm:nested} rather than on explicit constants. Second, the construction is moment-based: A2 must hold, the apparatus is silent for laws without the required moments, and Table~\ref{tab:e2} shows that at the edge ($t_5$) the feasible estimator degrades gracefully but no longer targets a defined constant; a characteristic-function reformulation is the appropriate tool there. A second route keeps the polynomial form but leaves the integer-power basis: a signed-parity fractional-power family indexed by a continuous parameter, for which a degree-two variance-reduction coefficient has been derived in closed form \citep{zabolotnii2026patp}, trades finite integer moments of order $2S$ for fractional ones. Whether it recovers any part of the degree-three reserve under A1, and at what conditioning cost, is open; we did not test it here. Third, $1/g_S$ is a within-polynomial-class frontier, not a semiparametric bound over all procedures. For bounded-support laws such as the uniform the regular location Fisher information is not the relevant finite constant, for Gamma-shape laws at or below shape two it diverges at the boundary, and where the regularity conditions of \citet{bickel1982adaptive} and \citet{kim2019semiparametric} hold, adaptive and density-based procedures target information beyond any finite-degree span. Table~\ref{tab:e1-sym} therefore establishes a separation between PMM$_3$ and the degree-two family, not superiority over broader semiparametric procedures, which we did not implement. Fourth, the SLS $=$ PMM$_2$ identity is asymptotic and specific to linear regression under A1: at finite $n$ the two estimators differ at order $O(1/n)$, and outside A1 Proposition~\ref{prop:hetero} governs. Fifth, the reserve is a large-sample statement with a law-dependent finite-sample behaviour. Proposition~\ref{prop:feasible} gives the feasible estimator the oracle variance to first order, but the second-order term is not small at $n\le200$: on skewed laws it \emph{raises} the realized efficiency above $1/g_3$, on platykurtic laws it lowers it, and $\operatorname{sd}(\hat\gamma_6)$ is of the order of $\gamma_6$ for $n\le500$ on the Gamma and lognormal laws. No regularization of the body we tried improved on the plain plug-in; the pretest of \S\ref{subsec:feasible} is the tool we can recommend, and its cost is that it forgoes the reserve at $n\le100$ on laws where the plain PMM$_3$ would already have collected part of it.

\section{Conclusion}
\label{sec:conclusion}

For linear regression with i.i.d., conditionally homoskedastic errors, optimally weighted second-order least squares and degree-two polynomial maximization are one estimating equation: the same optimal combination of $\{e,e^2-\sigma^2\}$, the Kunchenko body as the SLS weighting operator, one influence function, and the common slope variance $c_2g_2/N$. The identity stops exactly where the optimal slope combination starts to depend on the design. Within the polynomial moment class the next degree holds a reserve, quantified in closed form, that the second-order span cannot read and that any efficient estimator on the degree-three span can; the feasible plug-in estimator attains it to first order under sixth moments, a bootstrap pretest decides when to use it, and it survives on eleven simulated laws, two public data sets, and a drilling-calibrated error law. Three directions follow: explicit constants and the estimation trade-off for $S\ge4$; the nonlinear-regression case, where the instrument $\dot\mu_i(\theta)$ replaces the design and Corollary~\ref{cor:re}'s invariance is lost; and a conditional body $F_S(x)$ that would carry the degree axis into the heteroskedastic setting where conditional SLS now stands alone.

\section*{Supplementary material}
A supplementary document (PDF) accompanies the paper: S1 related estimators; S2 full proofs of Theorem~\ref{thm:sls-pmm2}, Proposition~\ref{prop:hetero} and Proposition~\ref{prop:blind}; S3 proof of Proposition~\ref{prop:feasible}; S4 the Lean~4 development; S5 complete Monte Carlo tables; S6 data-processing details.

\section*{Data and code availability}
The Lean~4 proofs, all R scripts, committed result files, and step-by-step reproduction instructions are released at \url{https://github.com/SZabolotnii/PMM-SLS-BRIDGE-code-supplement}. The \textsf{EstemPMM} package is on CRAN at \url{https://cran.r-project.org/package=EstemPMM}. The \emph{cars} and \emph{faithful} data ship with R; the FORGE well 58-32 log is available from the Geothermal Data Repository under \url{https://doi.org/10.15121/1495411}.

%% file: supplement_aism.tex
\begin{abstract}
This document accompanies the main paper and is organized as follows. Section~S1 discusses related estimators omitted from the main text. Section~S2 gives the full proofs of Theorem~1, Corollary~1, Proposition~1 and Proposition~2 of the main paper. Section~S3 proves Proposition~3 (the feasible plug-in estimator). Section~S4 describes the Lean~4 development and maps each formal statement to the main text. Section~S5 collects the complete Monte Carlo tables, including the pilot tables of the original submission and the full grids of the revision experiments. Section~S6 documents the two data examples and the drilling-calibrated error law. Equation, table and theorem numbers with the prefix ``S'' refer to this document; unprefixed numbers refer to the main paper. Notation and Assumptions A1--A4 are those of the main paper.
\end{abstract}

\section{Related estimators}
\label{sec:s1}

PMM and SLS sit inside a broader family of methods that exploit moment information without committing to a full likelihood. Quadratic inference functions \citep{qu2000qif} also optimize a quadratic form generated by basis components, but for correlated-data estimating equations rather than along the residual-power degree axis studied in the main paper. Huber's $M$-estimators \citep{huber1981robust} bound the influence of outliers through a fixed $\psi$-function, trading a small efficiency loss at the model for robustness against contamination; an $M$-estimator fixes its $\psi$-function a priori, whereas efficient GMM, PMM (through the body $F_S$) and SLS (through $W_i$) read the optimal weighting from the estimated moment structure of the data. Hosking's L-moments \citep{hosking1990lmoments} reparameterize a distribution through linear combinations of order statistics, giving low-variance shape summaries that are mainly descriptive rather than tuned to a regression slope. Robustness and descriptive parsimony, rather than slope efficiency under a known moment structure, are the aims of these two, so they complement rather than compete with the estimators compared in the main paper. The higher-order-statistics tradition \citep{mendel1991tutorial} shares PMM's premise that cumulants beyond the second carry usable information, but works chiefly in the spectral (polyspectral) domain.

Within the estimating-function literature, the projection fact that a minimum-variance combination of a given set of moment functions is unique is classical \citep{godambe1960optimum,crowder1987linear,godambe1989extension,heyde1997quasi}, and its GMM and optimal-instrument forms are in \citet{hansen1982large}, \citet{chamberlain1987asymptotic} and \citet{newey1993efficient}. The main paper claims no novelty for this fact; what it adds is the identification of the two concrete constructions on the span $\{e,e^2-\sigma^2\}$ and the quantification of the degree-three span above it.

\section{Full proofs for Section~3 and Section~4.4 of the main paper}
\label{sec:s2}

Throughout, $m(e)=(e,e^2-\sigma^2)^{\!\top}$, $F_2=\E[m(e)m(e)^{\!\top}]$, $b=(1,0)^{\!\top}$, $\rho_i=(y_i-\mu_i,\,y_i^2-\mu_i^2-\sigma^2)^{\!\top}$, and $T_i=\bigl(\begin{smallmatrix}1&0\\2\mu_i&1\end{smallmatrix}\bigr)$, so that $\rho_i=T_im(e_i)$.

\subsection{Proof of Theorem~1}

\emph{Step 1 (both reduce to a quadratic program in $\mathcal M$).}
Any estimator built from $\mathcal M=\operatorname{span}\{m_1,m_2\}$ uses an estimating equation $\sum_ix_i\,a^{\!\top}m(e_i)=0$ for some $a\in\R^2$. Standard M-estimation theory under A1--A4 gives the slope influence function $\psi_a(e,x)=-Q_a^{-1}x\,a^{\!\top}m(e)$ with Jacobian factor $Q_a=\E[xx^{\!\top}]\,(a^{\!\top}b)$, so that the slope asymptotic variance factor is
\begin{equation}
V(a)=c_2\,\frac{a^{\!\top}F_2a}{(a^{\!\top}b)^2}.
\label{eq:s-vfactor}
\end{equation}
The Cauchy--Schwarz inequality in the inner product $\langle u,v\rangle=u^{\!\top}F_2v$ gives $(a^{\!\top}b)^2=\langle a,F_2^{-1}b\rangle^2\le(a^{\!\top}F_2a)(b^{\!\top}F_2^{-1}b)$, hence $V(a)\ge c_2/(b^{\!\top}F_2^{-1}b)$, with equality iff $a\propto F_2^{-1}b$. The variance-minimizing combination over $\mathcal M$ is therefore unique up to scale and equal to $h^\star=F_2^{-1}b$. This is the PMM degree-two construction: minimizing the variance of the degree-$S$ stochastic polynomial estimating equation subject to unit location sensitivity is the constrained quadratic program $\min_hh^{\!\top}F_Sh$ subject to $h^{\!\top}b=1$, whose Lagrange stationarity condition is $F_Sh\propto b$ \citep{kunchenko2002polynomial}; at $S=2$ this is the normal system of the theorem.

\emph{Step 2 (SLS realizes the same optimum).}
The SLS first-order conditions for the slope are $\sum_i\dot\rho_i^{\!\top}W_i\rho_i=0$ with $\dot\rho_i=\partial\rho_i/\partial\beta_1$. \citet{wang2008second} show that $W_i^\star=\{\E[\rho_i\rho_i^{\!\top}\mid x_i]\}^{-1}$ minimizes the asymptotic variance within the class of SLS weights; \citet{kim2012efficiency} show that the resulting estimator is efficient in the model that fixes the first two conditional moments. Since $\rho_i=T_im(e_i)$ with $\det T_i=1$, the conditional second moment factorizes as $\E[\rho_i\rho_i^{\!\top}\mid x_i]=T_iF_2T_i^{\!\top}$ under A1, so
\begin{equation}
W_i^\star=T_i^{-\top}F_2^{-1}T_i^{-1}.
\label{eq:s-wstar}
\end{equation}
The slope Jacobian is $\dot\rho_i=-x_i(1,2\mu_i)^{\!\top}$, and
\begin{equation}
T_i^{-1}\begin{pmatrix}1\\2\mu_i\end{pmatrix}=\begin{pmatrix}1&0\\-2\mu_i&1\end{pmatrix}\begin{pmatrix}1\\2\mu_i\end{pmatrix}=\begin{pmatrix}1\\0\end{pmatrix}=b,
\qquad\text{so}\qquad (1,2\mu_i)\,T_i^{-\top}=b^{\!\top}.
\label{eq:s-tcancel}
\end{equation}
Substituting \eqref{eq:s-wstar}, \eqref{eq:s-tcancel} and $\rho_i=T_im(e_i)$,
\[
\dot\rho_i^{\!\top}W_i^\star\rho_i=-x_i\,(1,2\mu_i)\,T_i^{-\top}F_2^{-1}T_i^{-1}T_im(e_i)=-x_i\,(F_2^{-1}b)^{\!\top}m(e_i).
\]
The slope estimating equation therefore collapses to $\sum_ix_i(F_2^{-1}b)^{\!\top}m(e_i)=0$, up to the scalar normalization of Step~1. SLS selects $a=h^\star=F_2^{-1}b$, and its optimal weighting is the PMM body inverse, $W^\star\propto F_2^{-1}$.

\emph{Step 3 (identical influence function and variance).}
With a common $h^\star$ the influence functions of Step~1 coincide and equal $\psi(e,x)=-Q^{-1}x(h^\star)^{\!\top}m(e)$ with $Q=\E[xx^{\!\top}]\,b^{\!\top}h^\star$; first-order asymptotic equivalence yields a common limiting Gaussian law with slope variance factor $V(h^\star)=c_2/(b^{\!\top}F_2^{-1}b)$. Standardizing $e$ to unit variance and inverting the $2\times2$ matrix $F_2$ with $b=(1,0)^{\!\top}$, $b^{\!\top}F_2^{-1}b=(F_2^{-1})_{11}=(F_2)_{22}/\det F_2$. In standardized cumulants $\det F_2=\sigma^6[(2+\gamma_4)-\gamma_3^2]$ and $(F_2)_{22}=\sigma^4(2+\gamma_4)$, so
\begin{equation}
g_2=\frac{1}{\sigma^2\,b^{\!\top}F_2^{-1}b}=1-\frac{\gamma_3^2}{2+\gamma_4},
\label{eq:s-g2}
\end{equation}
and $V(h^\star)=c_2g_2$; dividing by $N$ gives the common asymptotic variance $c_2g_2/N$. Since $\gamma_3^2\ge0$ and $2+\gamma_4=\Var(e^2)/\sigma^4>0$ for any non-degenerate law, $0<g_2\le1$ with $g_2=1$ iff $\gamma_3=0$. \qed

\subsection{Proof of Corollary~1 (design invariance)}

The conditional second-moment matrix of $\rho_i$ has determinant $\Delta=\sigma^2\mu_4-\sigma^6-\kappa_3^2=\sigma^6(2+\gamma_4-\gamma_3^2)$, and the slope-information quadratic form $q=g_{\beta_0}^{\!\top}\Sigma_i^{-1}g_{\beta_0}=(\mu_4-\sigma^4)/\Delta$; in both expressions the $\mu_i$ and $\mu_i^2$ terms cancel identically, which is the algebraic content of \eqref{eq:s-tcancel} and is machine-checked (Section~S4). The slope is then orthogonal to the intercept and to $\sigma^2$ for a centered regressor, and a direct computation gives $\Var(\hat\beta_1)=\sigma^2g_2/(Nv_x)$ with $v_x=\Var(x)$, so that $v_x$ cancels against the OLS variance $\sigma^2/(Nv_x)$ and the slope attains exactly $\RE=1/g_2$. For centered $\chi^2(3)$ ($\gamma_3^2=8/3$, $\gamma_4=4$) and centered $\mathrm{Gamma}(2,1)$ ($\gamma_3^2=2$, $\gamma_4=3$) this gives $1/g_2=9/5=1.80$ and $5/3=1.667$; both evaluations are also checked in Lean. An independent asymptotic sandwich computation for the $\chi^2(3)$ case returned $1.798$. \qed

\subsection{Proof of Proposition~1 (scope of the equivalence)}

The factorization $\E[\rho_i\rho_i^{\!\top}\mid x_i]=T_iF_2(x_i)T_i^{\!\top}$ is immediate from $\rho_i=T_im(e_i)$ and $\E[m(e_i)m(e_i)^{\!\top}\mid x_i]=F_2(x_i)$. In Step~2 above the maps $T_i$ cancel against $W_i^\star=T_i^{-\top}F_2(x_i)^{-1}T_i^{-1}$ to leave the SLS vector field on the left of
\[
\sum_ix_i\,\bigl(F_2(x_i)^{-1}b\bigr)^{\!\top}m(e_i)=0
\qquad\text{versus}\qquad
\sum_ix_i\,\bigl(F_2^{-1}b\bigr)^{\!\top}m(e_i)=0,
\]
the unconditional degree-two PMM equation on the right; the two are the same vector field iff $F_2(x_i)^{-1}b$ does not depend on $x_i$, in particular whenever the conditional body is constant (A1), but more generally whenever the first column of $F_2(x)^{-1}$ is $x$-invariant.

(i) Under symmetry $\kappa_3(x)\equiv0$ makes every $F_2(x)$ diagonal, so $F_2(x)^{-1}b=(1/\sigma^2(x),0)^{\!\top}$ and the conditional SLS slope equation is $\sum_ix_ie_i/\sigma^2(x_i)=0$, the weighted least squares normal equation, whose asymptotic variance is the Aitken bound, strictly below the OLS variance unless $\sigma^2(x)$ is constant. The unconditional body $F_2=\E[F_2(x)]$ is also diagonal, so $F_2^{-1}b=(1/\bar\sigma^2,0)^{\!\top}$ and the unconditional PMM$_2$ (and constant-weight SLS) slope equation is $\sum_ix_ie_i=0$, ordinary least squares.

(ii) The unconditional-body slope estimating function is $\sum_ix_i\,(h_1e_i+h_2(e_i^2-\bar\sigma^2))$ with $h=F_2^{-1}b$. Its population mean is $h_1\E[xe]+h_2\E[x(e^2-\bar\sigma^2)]=h_2\,\E[x(\sigma^2(x)-\bar\sigma^2)]=h_2\operatorname{Cov}(x,\sigma^2(x))$, using $\E[xe]=\E[x\E[e\mid x]]=0$ and $\E[xe^2]=\E[x\sigma^2(x)]$. For asymmetric errors $h_2\propto-\kappa_3\neq0$, so the mean is non-zero whenever $\operatorname{Cov}(x,\sigma^2(x))\neq0$, giving an asymptotic bias. The same computation with $\bar\sigma^2$ replaced by a constant $\sigma^2$ shows that the constant-$\sigma^2$ SLS second residual has conditional mean $\sigma^2(x)-\sigma^2\neq0$; restoring $\E[\rho_2\mid x]=0$ requires the variance model $\sigma^2(x;\tau)$ of \citet{wang2008second}. \qed

\subsection{Proof of Proposition~2 (symmetric-kurtosis blindness of the second-order span)}

The optimally weighted SLS slope first-order condition is $\sum_i(\partial_\beta\rho_i)^{\!\top}W_i\rho_i=0$ with $\partial_\beta\rho_i=-x_i(1,2\mu_i)^{\!\top}$. By \eqref{eq:s-tcancel} this reduces to $\sum_ix_i\,(F_2^{-1}b)^{\!\top}m(e_i)=0$, the degree-two PMM equation with optimal combination $a=F_2^{-1}b$. The off-diagonal of $F_2$ is $\kappa_3=\mu_3$. Under symmetry $\mu_3=0$, so $F_2$ is diagonal and $a=(1/\sigma^2,0)^{\!\top}$: the optimal combination places zero weight on the squared residual and the slope condition collapses to $\sum_ix_ie_i=0$, the OLS normal equation.

The conditional cross-moment is not the mechanism: $\E[\rho_{i1}\rho_{i2}\mid x_i]=2\mu_i\sigma^2+\mu_3$ does not vanish at $\mu_3=0$, so $W_i$ is not block-diagonal. Rather, the optimal weight annihilates $\rho_2$ in the slope direction, $W_i^\star\partial_\beta\rho_i\propto T_i^{-\top}F_2^{-1}b=(1/\sigma^2,0)^{\!\top}$, whose second component is zero. No even cumulant $\gamma_4,\gamma_6,\dots$ enters, because the centered basis $m=(e,e^2-\sigma^2)$ contains no monomial above $e^2$. Hence the SLS slope variance equals the OLS variance, $g_2=1$, independently of $\gamma_4$. \qed

\subsection{Proof details for Theorem~4}

With $\sigma=1$ the central moments are $m_2=1$, $m_3=\gamma_3$, $m_4=\gamma_4+3$, $m_5=\gamma_5+10\gamma_3$, $m_6=\gamma_6+15\gamma_4+10\gamma_3^2+15$, and the body entries are $F_{jk}=m_{j+k}-m_jm_k$ with $m_1=0$: $F_{11}=1$, $F_{12}=\gamma_3$, $F_{13}=\gamma_4+3$, $F_{22}=m_4-1=\gamma_4+2$, $F_{23}=m_5-m_2m_3=\gamma_5+9\gamma_3$, $F_{33}=m_6-m_3^2=\gamma_6+15\gamma_4+9\gamma_3^2+15$; the sensitivity vector is $b=(1,2m_1,3m_2)^{\!\top}=(1,0,3)^{\!\top}$. Expanding $\det F_3$ and $b^{\!\top}\operatorname{adj}(F_3)b=\operatorname{adj}_{11}+6\operatorname{adj}_{13}+9\operatorname{adj}_{33}$ gives the polynomials $N$ and $D$ of the main text; the two identities, together with $D-N$, the dependence of $D$ on $\gamma_5$ and $\gamma_6$, and the symmetric factorization $D=(\gamma_4+2)(6+9\gamma_4+\gamma_6)$, $D-N=\gamma_4^2(\gamma_4+2)$, are proved in Lean by \texttt{ring}-normalization (Section~S4). \qed

\section{Proof of Proposition~3 (feasible plug-in estimator)}
\label{sec:s3}

\emph{Setting.} Under A1--A4, let $\hat\beta^{(0)}$ be the OLS estimator, $\hat e_i=y_i-x_i^{\!\top}\hat\beta^{(0)}$, $\hat\alpha_k=N^{-1}\sum_i\hat e_i^k$ ($k\le2S$), $\alpha_k=\E[e^k]$, $\hat h=h(\hat\alpha)$ and $h=h(\alpha)$, where $h(\alpha)=F_S(\alpha)^{-1}b(\alpha)$ with $F_S(\alpha)_{jk}=\alpha_{j+k}-\alpha_j\alpha_k$ and $b(\alpha)_k=k\alpha_{k-1}$ ($\alpha_0=1$). Write $m_S(e)=(e,\dots,e^S)^{\!\top}$, $\bar m_S(e)=m_S(e)-\alpha_{1:S}$, and $z_i=x_i-\bar x$ for the centered regressor (a scalar in the model of the main paper; the argument is coordinatewise in general). The feasible slope estimator $\hat\beta_{S,1}$ solves, jointly with the intercept equation,
\begin{equation}
\Psi_N(\beta,\hat h)=\frac1N\sum_{i=1}^Nz_i\,\hat h^{\!\top}\bigl(m_S(y_i-x_i^{\!\top}\beta)-\hat\alpha_{1:S}\bigr)=\frac1N\sum_{i=1}^Nz_i\,\hat h^{\!\top}m_S(y_i-x_i^{\!\top}\beta)=0,
\label{eq:s-feasible}
\end{equation}
where the second equality holds identically in $\beta$ because $\sum_iz_i=0$.

\emph{Step 1 (the estimated means drop out).} Equation \eqref{eq:s-feasible} shows that $\hat\alpha_{1:S}$ does not enter the slope equation at all; the plug-in enters only through $\hat h$. This is exact, not asymptotic, and is the reason the intercept is required in A4.

\emph{Step 2 (consistency of $\hat h$ under $2S$-th moments).} Expand $\hat e_i^k=(e_i-x_i^{\!\top}\delta)^k$ with $\delta=\hat\beta^{(0)}-\beta=O_p(N^{-1/2})$:
\[
\hat\alpha_k-\alpha_k=\Bigl(\frac1N\sum_ie_i^k-\alpha_k\Bigr)+\sum_{j=1}^k\binom kj(-1)^j\,\frac1N\sum_ie_i^{k-j}(x_i^{\!\top}\delta)^j .
\]
The first bracket is $o_p(1)$ by the law of large numbers under $\E|e|^{2S}<\infty$ (A2), since $k\le2S$. For the $j$-th remainder, $|N^{-1}\sum_ie_i^{k-j}(x_i^{\!\top}\delta)^j|\le\|\delta\|^jN^{-1}\sum_i|e_i|^{k-j}\|x_i\|^j$, and by independence of $e$ and $x$ (A1) and H\"older, $\E[|e|^{k-j}\|x\|^j]\le(\E|e|^{2S})^{(k-j)/(2S)}(\E\|x\|^{2S})^{j/(2S)}<\infty$ by A2 and A4, so the average is $O_p(1)$ and the term is $O_p(N^{-j/2})$. Hence $\hat\alpha_k\to_p\alpha_k$ for every $k\le2S$. On the open set $\{\alpha:\det F_S(\alpha)>0\}$ the map $\alpha\mapsto h(\alpha)$ is a rational function (Cramer's rule) and therefore continuous; A3 places $\alpha$ in this set, so $\hat h\to_ph$ by the continuous mapping theorem. A3 also makes $h$ the unique solution of the normal system, so the selected polynomial direction is well defined.

\emph{Step 3 (the plug-in does not affect the first-order law).} Let $\Psi_N(\beta,h)$ denote \eqref{eq:s-feasible} with the population $h$. At the true $\beta$,
\[
\Psi_N(\beta,\hat h)=\Psi_N(\beta,h)+(\hat h-h)^{\!\top}U_N,\qquad U_N=\frac1N\sum_{i=1}^Nz_i\,m_S(e_i).
\]
Under A1, $\E[z_im_S(e_i)]=\E[z_i]\,\E[m_S(e_i)]=0$ for the population-centered regressor, and the sample centering changes $U_N$ by $-\bar z\,N^{-1}\sum_im_S(e_i)=O_p(N^{-1/2})\cdot O_p(1)$; with finite second moments of $z_im_S(e_i)$, which hold under A2 and A4 ($\E[z^2e^{2k}]=\E[z^2]\E[e^{2k}]<\infty$ for $k\le S$), the central limit theorem gives $U_N=O_p(N^{-1/2})$. Hence $(\hat h-h)^{\!\top}U_N=o_p(1)\cdot O_p(N^{-1/2})=o_p(N^{-1/2})$: the estimating function with the plug-in coefficients differs from the oracle estimating function by a term of smaller order than the oracle function's own fluctuation. This is the standard situation of an estimated nuisance quantity whose expected derivative is zero \citep[Section~6.1]{newey1994large}; here it is exact rather than approximate because the estimating function is linear in $h$.

\emph{Step 4 (Jacobian and conclusion).} The derivative of $\Psi_N(\beta,\hat h)$ with respect to the slope at the truth is $-N^{-1}\sum_iz_ix_{i}\,\hat h^{\!\top}m_S'(e_i)$, where $m_S'(e)=(1,2e,\dots,Se^{S-1})^{\!\top}$; by the law of large numbers, A1 and $\hat h\to_ph$ it converges in probability to $-\Var(x)\,h^{\!\top}b$ with $b=\E[m_S'(e)]=(1,2\alpha_1,\dots,S\alpha_{S-1})^{\!\top}$, and $h^{\!\top}b=b^{\!\top}F_S^{-1}b>0$ by A3. A uniform law of large numbers on a neighbourhood of $\beta$ (available under A2 and A4, since $m_S$ is a polynomial) and the usual M-estimation expansion \citep[Theorems~2.1 and~3.1]{newey1994large} then give consistency of $\hat\beta_{S,1}$ and
\begin{align*}
\sqrt N(\hat\beta_{S,1}-\beta_1)&=\bigl(\Var(x)\,h^{\!\top}b\bigr)^{-1}\sqrt N\,h^{\!\top}U_N+o_p(1)\\
&\xrightarrow{d}\mathcal N\Bigl(0,\ \frac{\Var(x)\,h^{\!\top}F_Sh}{\Var(x)^2(h^{\!\top}b)^2}\Bigr)=\mathcal N\Bigl(0,\ \frac{1}{\Var(x)\,b^{\!\top}F_S^{-1}b}\Bigr),
\end{align*}
using $h^{\!\top}F_Sh=h^{\!\top}b=b^{\!\top}F_S^{-1}b$. Since $c_2g_S=\sigma^2/\Var(x)\cdot1/(\sigma^2b^{\!\top}F_S^{-1}b)$, this is $\mathcal N(0,c_2g_S)$, the oracle law. A single Newton step from the $\sqrt N$-consistent OLS start yields the same expansion; iterating and re-estimating $\hat\alpha$ at the updated $\hat\beta$ changes $\hat h$ by $o_p(1)$ and hence nothing at first order.

\emph{Remark (where the $4S$-th moment enters).} Step~2 used only $\hat h-h=o_p(1)$. If one wants a $\sqrt N$-rate and a limit law for $\hat h$ or for $\hat g_S=1/(\hat\alpha_2\,\hat b^{\!\top}\widehat F_S^{-1}\hat b)$, for instance to attach a confidence interval to the estimated reserve $1/\hat g_3-1/\hat g_2$ of the selection rule, the central limit theorem for $\hat\alpha_{2S}$ requires $\Var(e^{2S})<\infty$, i.e.\ $\E[e^{4S}]<\infty$ (the twelfth moment for $S=3$), and the delta method applies to the rational map $h(\cdot)$. The bootstrap used in the pretest of the main paper estimates the same dispersion without this moment condition entering the first-order law of $\hat\beta_S$. The second-order term of the variance of $\hat\beta_S$, which is what the finite-sample tables of Section~S5 measure, likewise depends on higher moments; it is not characterized here. \qed

\section{The Lean~4 development}
\label{sec:s4}

The formalization is a single Lean~4 file (\texttt{PmmSls.lean}, Lean toolchain \texttt{v4.30.0}, Mathlib as pinned in the repository's \texttt{lake-manifest.json}) in the code supplement; \texttt{lake build} checks it. Its statements are about real numbers and polynomial identities, not about random variables: the moments enter as real parameters, and the hypotheses are the algebraic counterparts of A2--A3 (positivity of the variance, non-vanishing or positivity of the body factor). Table~\ref{tab:s-lean} maps each formal statement to the main paper. Nothing about the operator identity $W^\star\propto F_2^{-1}$, the influence functions, the asymptotic variance arguments, the nesting theorem, the heteroskedastic boundary or Proposition~3 is formalized; these are the analytic proofs of Section~S2, Section~S3 and the main text.

\begin{table}[ht]
\centering
\caption{Formal statements in \texttt{PmmSls.lean} and their counterparts in the main paper.}
\label{tab:s-lean}
\scriptsize
\setlength{\tabcolsep}{3pt}
\begin{tabular}{@{}p{4.0cm}p{5.3cm}l@{}}
\toprule
Lean statement & Content & Main paper\\
\midrule
\texttt{g2\_forms} & equivalence of the central-moment form $1-\kappa_3^2/(\sigma^2(\kappa_4+2\sigma^4))$ and the standardized form $1-\gamma_3^2/(2+\gamma_4)$ of $g_2$ & eq.~(4), Thm~1\\
\texttt{g2\_le\_one} & $g_2\le1$ when $\sigma^2>0$ and $2\sigma^4+\kappa_4>0$ & Thm~1\\
\texttt{g2\_eq\_one\_iff\_symmetric} & $g_2=1\iff\kappa_3=0$ under the same hypotheses & Thm~1, Prop.~2\\
\texttt{det\_cond\_cov\_mu\_free} & $\det\E[\rho\rho^{\!\top}\mid x]$ does not depend on $\mu(x)$ & Cor.~1\\
\texttt{slope\_info\_numerator\_mu\_free} & the slope-information quadratic form does not depend on $\mu(x)$ & Cor.~1\\
\texttt{slope\_RE\_eq\_g2}, \texttt{g2\_times\_slope\_RE\_eq\_one} & the design-invariant slope efficiency equals $1/g_2$ & Cor.~1\\
\texttt{RE\_chisq3}, \texttt{RE\_gamma2} & $1/g_2=9/5$ and $5/3$ for the two reference laws & \S3.2\\
\texttt{g3num}, \texttt{g3den} & the polynomials $N$ and $D$ & eqs.~(19)--(20)\\
\texttt{det\_F3\_eq\_g3num} & $\det F_3=N$ for the standardized body & Thm~4\\
\texttt{g3den\_eq\_cramer} & $b^{\!\top}\operatorname{adj}(F_3)b=D$ & Thm~4\\
\texttt{g3\_reduction\_numerator} & the expression for $D-N$ & Thm~4\\
\texttt{g3den\_gamma5\_sensitivity}, \texttt{g3den\_gamma6\_sensitivity} & $D$ depends on $\gamma_5$ and on $\gamma_6$ & Thm~4\\
\texttt{g3\_symmetric\_reduction} & at $\gamma_3=\gamma_5=0$, $N/D=1-\gamma_4^2/(6+9\gamma_4+\gamma_6)$ & Thm~3, eq.~(16)\\
\bottomrule
\end{tabular}
\end{table}

\section{Complete Monte Carlo tables}
\label{sec:s5}

All regression experiments use the design of the main paper, $y_i=2+1.5x_i+e_i$, $x_i\sim\mathrm U(0,5)$, slope as estimand, common random numbers within a cell, and report $\RE=\mathrm{MSE}_{\OLS}/\mathrm{MSE}$. Standard errors in parentheses are paired bootstrap standard errors over replications ($500$ resamples). Tables~\ref{tab:s-h1}--\ref{tab:s-gauss} are the tables of the original submission (pilot regression cells with $M=500$; oracle cells with $M=10\,000$); Tables~\ref{tab:s-e1}--\ref{tab:s-hetero} are the revision experiments. The scripts are \texttt{asym\_s3\_mc.R} and \texttt{hetero\_mc.R} (directory \texttt{experiments/}) and \texttt{e1\_gmm3\_competitor.R}, \texttt{e2\_feasible\_sweep.R}, \texttt{e3\_m6\_shrinkage\_selection.R} (directory \texttt{experiments/revision\_2026-09/}) in the code supplement.

\subsection{Tables of the original submission}

\begin{table}[ht]
\centering
\caption{Pilot regression Monte Carlo under asymmetric errors ($M=500$): \emph{feasible} $\RE$ over OLS for PMM$_2$, SLS and GMM$_2$, the asymptotic $1/g_2$, and the per-replication correlation between the PMM$_2$ and SLS slopes. Figure~1 of the main paper is drawn from this table.}
\label{tab:s-h1}
\small
\begin{tabular}{@{}llccccc@{}}
\toprule
Error & $n$ & PMM$_2$ & SLS & GMM$_2$ & $1/g_2$ & cor(PMM$_2$,SLS)\\
\midrule
$\chi^2(3)$
 & 100  & 1.871 & 1.859 & 1.812 & 1.80 & 0.9959 \\
 & 200  & 1.884 & 1.875 & 1.859 & 1.80 & 0.9984 \\
 & 500  & 1.871 & 1.864 & 1.871 & 1.80 & 0.9994 \\
 & 1000 & 1.828 & 1.824 & 1.831 & 1.80 & 0.9997 \\
\midrule
$\mathrm{Gamma}(2,1)$
 & 100  & 1.708 & 1.697 & 1.669 & 1.67 & 0.9967 \\
 & 200  & 1.735 & 1.725 & 1.727 & 1.67 & 0.9985 \\
 & 500  & 1.726 & 1.719 & 1.744 & 1.67 & 0.9995 \\
 & 1000 & 1.684 & 1.683 & 1.693 & 1.67 & 0.9997 \\
\bottomrule
\end{tabular}
\end{table}

\begin{table}[ht]
\centering
\caption{\emph{Oracle} location Monte Carlo ($M=10\,000$, population coefficients $h=F_S^{-1}b$, paired bootstrap standard errors): $\RE$ against the sample mean next to the asymptotic constants.}
\label{tab:s-oracle-loc}
\small
\begin{tabular}{@{}llcccc@{}}
\toprule
Error & $n$ & PMM$_2$ $\RE$ (SE) & $1/g_2$ & PMM$_3$ $\RE$ (SE) & $1/g_3$\\
\midrule
$\chi^2(3)$
 & 200  & $1.826\,(0.025)$ & 1.799 & $2.576\,(0.040)$ & 2.479 \\
 & 1000 & $1.808\,(0.025)$ & 1.799 & $2.528\,(0.041)$ & 2.479 \\
 & 4000 & $1.746\,(0.023)$ & 1.799 & $2.402\,(0.038)$ & 2.479 \\
\midrule
$\mathrm{Gamma}(2,1)$
 & 200  & $1.632\,(0.021)$ & 1.667 & $2.146\,(0.031)$ & 2.167 \\
 & 1000 & $1.657\,(0.020)$ & 1.667 & $2.167\,(0.032)$ & 2.167 \\
 & 4000 & $1.668\,(0.021)$ & 1.667 & $2.146\,(0.031)$ & 2.167 \\
\bottomrule
\end{tabular}
\end{table}

\begin{table}[ht]
\centering
\caption{\emph{Oracle} degree-three Monte Carlo on the four laws of Table~1 of the main paper ($M=10\,000$, population coefficients $h=F_3^{-1}b$, paired bootstrap standard errors): realized $\RE_3$ over OLS next to the closed-form $1/g_3$. In $1.2\times10^5$ fits the Newton step guard engaged twice and the divergence fallback never.}
\label{tab:s-oracle-asym}
\small
\begin{tabular}{lccc}
\toprule
error law & $1/g_3$ & $\RE_3$ ($n{=}1000$) & $\RE_3$ ($n{=}4000$)\\
\midrule
$\chi^2(2)=\mathrm{Exp}$ & $3.00$ & $2.99\,(0.05)$ & $3.00\,(0.05)$\\
$\mathrm{Gamma}(1.2,1)$ & $2.76$ & $2.75\,(0.04)$ & $2.75\,(0.05)$\\
$\chi^2(3)$ & $2.49$ & $2.53\,(0.04)$ & $2.40\,(0.04)$\\
$\mathrm{Gamma}(2,1)$ & $2.17$ & $2.17\,(0.03)$ & $2.15\,(0.03)$\\
\bottomrule
\end{tabular}
\end{table}

\begin{table}[ht]
\centering
\caption{Pilot regression Monte Carlo under symmetric platykurtic errors $\mathrm U(-\sqrt3,\sqrt3)$ ($M=500$): \emph{feasible} $\RE$ over OLS. PMM$_3$ here is the symmetric even-moment routine of \textsf{EstemPMM}. Asymptotic $1/g_3=3.33$.}
\label{tab:s-h2}
\small
\begin{tabular}{@{}lcccc@{}}
\toprule
$n$ & SLS & GMM$_2$ & PMM$_3$ & $1/g_3$\\
\midrule
50  & 0.910 & 0.951 & 2.065 & 3.33 \\
100 & 0.955 & 0.985 & 2.953 & 3.33 \\
200 & 0.990 & 1.002 & 2.872 & 3.33 \\
500 & 1.002 & 1.006 & 3.159 & 3.33 \\
\bottomrule
\end{tabular}
\end{table}

\begin{table}[ht]
\centering
\caption{Pilot Gaussian control ($M=500$): \emph{feasible} $\RE$ over OLS and the PMM$_2$--SLS correlation. No reserve exists; the sub-unity efficiency is the finite-sample cost of the extra moment equation.}
\label{tab:s-gauss}
\small
\begin{tabular}{@{}lccc@{}}
\toprule
$n$ & PMM$_2$ & SLS & cor(PMM$_2$,SLS)\\
\midrule
50  & 0.931 & 0.896 & 0.9958 \\
100 & 0.991 & 0.982 & 0.9990 \\
200 & 0.989 & 0.986 & 0.9999 \\
500 & 0.983 & 0.983 & 1.0000 \\
\bottomrule
\end{tabular}
\end{table}

\begin{figure}[ht]
  \centering
  \includegraphics[width=0.78\textwidth]{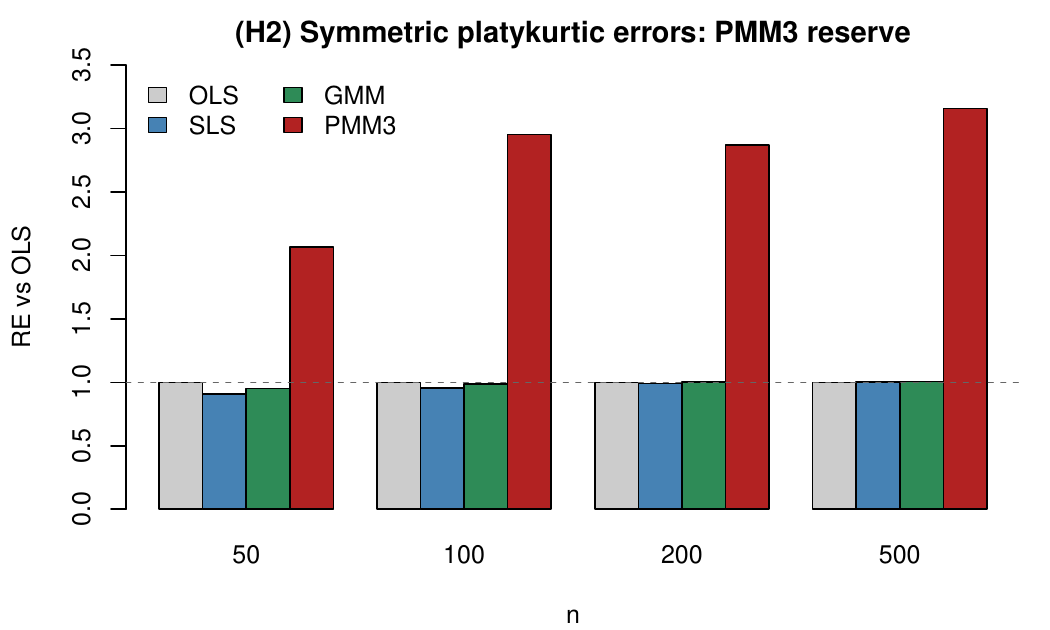}
  \caption{Pilot study ($M=500$, Table~\ref{tab:s-h2}): on symmetric platykurtic errors SLS and GMM$_2$ coincide with OLS while the symmetric degree-three estimator attains $\RE$ up to $3.16$.}
  \label{fig:s-h2}
\end{figure}

\begin{figure}[ht]
  \centering
  \includegraphics[width=0.62\textwidth]{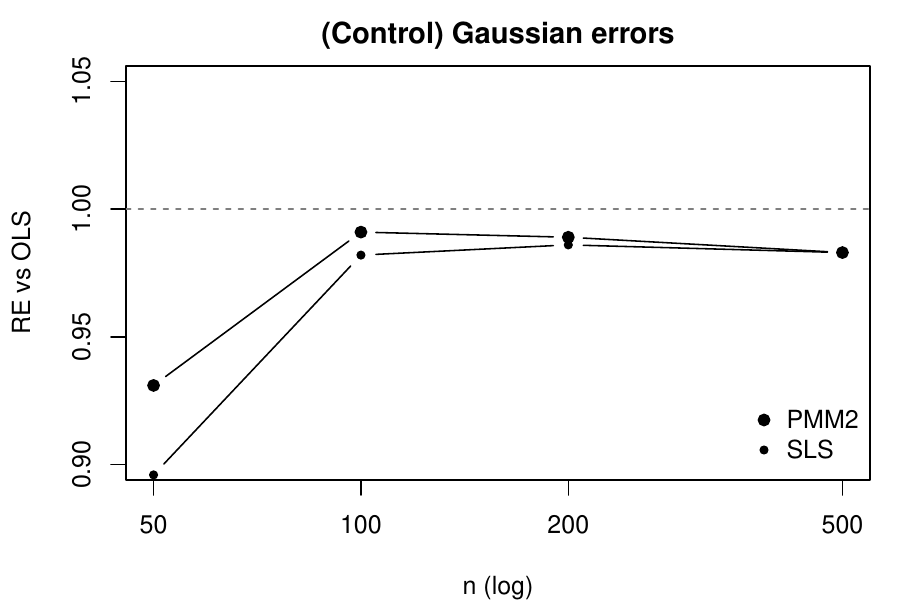}
  \caption{Pilot Gaussian control ($M=500$, Table~\ref{tab:s-gauss}): PMM$_2$ and SLS incur no asymptotic efficiency loss relative to OLS; the small finite-sample shortfall vanishes as $n$ grows.}
  \label{fig:s-gauss}
\end{figure}

\subsection{Revision experiment E1: degree-three GMM competitor}

\begin{sidewaystable}
\centering
\caption{Experiment E1 ($M=2000$): \emph{feasible} $\RE$ over OLS with paired bootstrap standard errors for all seven estimators and four laws. PMM$_3$ is the general degree-three plug-in with the full $3\times3$ body; PMM$_3^{\mathrm{sym}}$ is the symmetric even-moment routine of \textsf{EstemPMM}, reported on every law to document that it is \emph{not} appropriate for skewed errors. cor$_3$ and diff$_3$ are the per-replication correlation and the mean paired difference (SE) between the PMM$_3$ and GMM$_3$ slopes.}
\label{tab:s-e1}
\scriptsize
\begin{tabular}{@{}llcccccccc@{}}
\toprule
Error & $n$ & PMM$_2$ & SLS & GMM$_2$ & PMM$_3$ & PMM$_3^{\mathrm{sym}}$ & GMM$_3$ & cor$_3$ & diff$_3$\\
\midrule
$\chi^2(3)$ ($1/g_2=1.80$, $1/g_3=2.50$)
 & 50 & $2.12\,(0.07)$ & $2.05\,(0.06)$ & $1.82\,(0.06)$ & $3.05\,(0.12)$ & $1.28\,(0.03)$ & $1.52\,(0.05)$ & 0.686 & $0.0162\,(0.0033)$ \\
 & 100 & $1.97\,(0.06)$ & $1.97\,(0.06)$ & $1.91\,(0.05)$ & $2.96\,(0.11)$ & $1.28\,(0.03)$ & $2.12\,(0.06)$ & 0.833 & $0.0020\,(0.0015)$ \\
 & 200 & $1.90\,(0.06)$ & $1.89\,(0.06)$ & $1.89\,(0.06)$ & $3.01\,(0.12)$ & $1.25\,(0.03)$ & $2.66\,(0.09)$ & 0.929 & $0.0002\,(0.0006)$ \\
 & 500 & $1.80\,(0.05)$ & $1.79\,(0.05)$ & $1.80\,(0.05)$ & $2.71\,(0.10)$ & $1.21\,(0.02)$ & $2.62\,(0.08)$ & 0.978 & $0.0000\,(0.0002)$ \\
 & 1000 & $1.83\,(0.05)$ & $1.83\,(0.05)$ & $1.85\,(0.05)$ & $2.74\,(0.10)$ & $1.22\,(0.02)$ & $2.72\,(0.09)$ & 0.987 & $0.0002\,(0.0001)$ \\
\midrule
$\mathrm{Gamma}(2,1)$ ($1/g_2=1.67$, $1/g_3=2.17$)
 & 50 & $1.72\,(0.05)$ & $1.73\,(0.05)$ & $1.58\,(0.04)$ & $2.27\,(0.08)$ & $1.19\,(0.02)$ & $1.34\,(0.04)$ & 0.748 & $0.0027\,(0.0018)$ \\
 & 100 & $1.76\,(0.05)$ & $1.75\,(0.05)$ & $1.72\,(0.05)$ & $2.32\,(0.09)$ & $1.20\,(0.02)$ & $1.82\,(0.05)$ & 0.878 & $0.0003\,(0.0008)$ \\
 & 200 & $1.74\,(0.05)$ & $1.72\,(0.05)$ & $1.75\,(0.05)$ & $2.46\,(0.08)$ & $1.21\,(0.02)$ & $2.22\,(0.07)$ & 0.942 & $0.0002\,(0.0003)$ \\
 & 500 & $1.81\,(0.05)$ & $1.80\,(0.05)$ & $1.82\,(0.05)$ & $2.40\,(0.09)$ & $1.24\,(0.02)$ & $2.36\,(0.08)$ & 0.982 & $0.0002\,(0.0001)$ \\
 & 1000 & $1.68\,(0.05)$ & $1.68\,(0.05)$ & $1.69\,(0.05)$ & $2.22\,(0.07)$ & $1.18\,(0.02)$ & $2.22\,(0.07)$ & 0.992 & $0.0003\,(0.0001)$ \\
\midrule
$\mathrm U(-\sqrt3,\sqrt3)$ ($1/g_2=1.00$, $1/g_3=3.33$)
 & 50 & $0.93\,(0.01)$ & $0.89\,(0.01)$ & $0.94\,(0.01)$ & $2.18\,(0.08)$ & $2.08\,(0.06)$ & $2.00\,(0.07)$ & 0.927 & $-0.0013\,(0.0006)$ \\
 & 100 & $0.96\,(0.01)$ & $0.95\,(0.01)$ & $0.98\,(0.01)$ & $2.78\,(0.11)$ & $2.61\,(0.09)$ & $2.50\,(0.09)$ & 0.972 & $-0.0003\,(0.0002)$ \\
 & 200 & $0.99\,(0.00)$ & $0.99\,(0.01)$ & $0.99\,(0.01)$ & $2.99\,(0.11)$ & $2.91\,(0.10)$ & $2.83\,(0.10)$ & 0.988 & $0.0003\,(0.0001)$ \\
 & 500 & $0.99\,(0.00)$ & $0.99\,(0.00)$ & $1.00\,(0.00)$ & $3.18\,(0.12)$ & $3.15\,(0.12)$ & $3.15\,(0.12)$ & 0.995 & $0.0000\,(0.0000)$ \\
 & 1000 & $0.99\,(0.00)$ & $0.99\,(0.00)$ & $1.00\,(0.00)$ & $3.37\,(0.12)$ & $3.34\,(0.12)$ & $3.33\,(0.12)$ & 0.998 & $0.0000\,(0.0000)$ \\
\midrule
$\mathcal N(0,1)$ ($1/g_2=1.00$, $1/g_3=1.00$)
 & 50 & $0.96\,(0.01)$ & $0.94\,(0.01)$ & $0.94\,(0.01)$ & $0.82\,(0.02)$ & $0.93\,(0.01)$ & $0.77\,(0.02)$ & 0.934 & $0.0005\,(0.0009)$ \\
 & 100 & $0.98\,(0.01)$ & $0.98\,(0.01)$ & $0.96\,(0.01)$ & $0.93\,(0.01)$ & $0.96\,(0.01)$ & $0.86\,(0.01)$ & 0.967 & $-0.0001\,(0.0004)$ \\
 & 200 & $1.00\,(0.01)$ & $1.00\,(0.01)$ & $0.99\,(0.01)$ & $0.98\,(0.01)$ & $0.99\,(0.01)$ & $0.94\,(0.01)$ & 0.983 & $-0.0004\,(0.0002)$ \\
 & 500 & $1.00\,(0.00)$ & $1.00\,(0.00)$ & $1.00\,(0.01)$ & $0.99\,(0.01)$ & $0.99\,(0.00)$ & $0.98\,(0.01)$ & 0.994 & $0.0001\,(0.0001)$ \\
 & 1000 & $1.00\,(0.00)$ & $1.00\,(0.00)$ & $1.00\,(0.00)$ & $1.00\,(0.00)$ & $1.00\,(0.00)$ & $1.00\,(0.01)$ & 0.997 & $0.0000\,(0.0000)$ \\
\bottomrule
\end{tabular}
\end{sidewaystable}

\subsection{Revision experiment E2: feasible-sample sweep}

\begin{table}[ht]
\centering
\caption{Experiment E2 ($M=2000$), laws 1--6: \emph{feasible} $\RE$ over OLS with paired bootstrap standard errors for the general degree-two and degree-three plug-ins; PMM$_2^{\mathrm E}$ is the \textsf{EstemPMM} degree-two routine as a cross-check; $\operatorname{sd}(\hat\gamma_6)$ is the Monte Carlo standard deviation of the sixth standardized cumulant of the OLS residuals. Asymptotic constants in the law headers.}
\label{tab:s-e2a}
\scriptsize
\setlength{\tabcolsep}{3pt}
\begin{tabular}{@{}llcccc@{}}
\toprule
Error law & $n$ & PMM$_2$ & PMM$_2^{\mathrm E}$ & PMM$_3$ & $\operatorname{sd}(\hat\gamma_6)$\\
\midrule
\multicolumn{6}{@{}l}{$\mathrm{Gamma}(1,1)=\mathrm{Exp}$ ($\gamma_3=2.00$, $\gamma_4=5.96$, $\gamma_6=116.7$; $1/g_2=2.00$, $1/g_3=3.00$)}\\
 & 50 & $2.28\,(0.09)$ & 2.18 & $3.90\,(0.17)$ & 47.5 \\
 & 100 & $2.26\,(0.08)$ & 2.21 & $3.86\,(0.17)$ & 103.4 \\
 & 200 & $2.23\,(0.08)$ & 2.21 & $3.53\,(0.14)$ & 194.1 \\
 & 500 & $2.27\,(0.07)$ & 2.26 & $3.72\,(0.14)$ & 217.5 \\
 & 1000 & $2.09\,(0.06)$ & 2.09 & $3.63\,(0.14)$ & 207.7 \\
 & 2000 & $1.95\,(0.06)$ & 1.94 & $3.16\,(0.12)$ & 175.4 \\
\midrule
\multicolumn{6}{@{}l}{$\mathrm{Gamma}(2,1)$ ($\gamma_3=1.41$, $\gamma_4=2.98$, $\gamma_6=28.8$; $1/g_2=1.67$, $1/g_3=2.17$)}\\
 & 50 & $1.78\,(0.05)$ & 1.72 & $2.27\,(0.08)$ & 22.3 \\
 & 100 & $1.78\,(0.05)$ & 1.76 & $2.32\,(0.09)$ & 42.6 \\
 & 200 & $1.75\,(0.05)$ & 1.74 & $2.46\,(0.08)$ & 66.4 \\
 & 500 & $1.81\,(0.05)$ & 1.81 & $2.40\,(0.09)$ & 76.1 \\
 & 1000 & $1.68\,(0.05)$ & 1.68 & $2.22\,(0.07)$ & 66.0 \\
 & 2000 & $1.61\,(0.05)$ & 1.61 & $2.20\,(0.07)$ & 58.0 \\
\midrule
\multicolumn{6}{@{}l}{$\mathrm{Gamma}(4,1)$ ($\gamma_3=1.00$, $\gamma_4=1.50$, $\gamma_6=7.1$; $1/g_2=1.40$, $1/g_3=1.60$)}\\
 & 50 & $1.39\,(0.03)$ & 1.37 & $1.44\,(0.05)$ & 9.6 \\
 & 100 & $1.36\,(0.03)$ & 1.35 & $1.52\,(0.04)$ & 19.4 \\
 & 200 & $1.50\,(0.03)$ & 1.49 & $1.66\,(0.04)$ & 30.3 \\
 & 500 & $1.43\,(0.03)$ & 1.42 & $1.61\,(0.04)$ & 27.1 \\
 & 1000 & $1.42\,(0.03)$ & 1.42 & $1.61\,(0.04)$ & 25.0 \\
 & 2000 & $1.41\,(0.04)$ & 1.41 & $1.64\,(0.04)$ & 18.5 \\
\midrule
\multicolumn{6}{@{}l}{$\mathrm{Gamma}(8,1)$ ($\gamma_3=0.71$, $\gamma_4=0.76$, $\gamma_6=2.0$; $1/g_2=1.22$, $1/g_3=1.29$)}\\
 & 50 & $1.17\,(0.02)$ & 1.16 & $1.10\,(0.03)$ & 6.5 \\
 & 100 & $1.18\,(0.02)$ & 1.18 & $1.16\,(0.03)$ & 10.3 \\
 & 200 & $1.23\,(0.02)$ & 1.23 & $1.27\,(0.03)$ & 11.5 \\
 & 500 & $1.23\,(0.02)$ & 1.23 & $1.28\,(0.03)$ & 11.2 \\
 & 1000 & $1.25\,(0.02)$ & 1.25 & $1.32\,(0.03)$ & 13.4 \\
 & 2000 & $1.18\,(0.02)$ & 1.18 & $1.24\,(0.03)$ & 6.4 \\
\midrule
\multicolumn{6}{@{}l}{lognormal$(0,0.5)$ ($\gamma_3=1.75$, $\gamma_4=5.96$, $\gamma_6=252.5$; $1/g_2=1.63$, $1/g_3=1.99$)}\\
 & 50 & $1.79\,(0.06)$ & 1.74 & $2.09\,(0.07)$ & 30.7 \\
 & 100 & $1.79\,(0.06)$ & 1.77 & $2.19\,(0.08)$ & 148.6 \\
 & 200 & $1.81\,(0.05)$ & 1.80 & $2.30\,(0.08)$ & 284.8 \\
 & 500 & $1.75\,(0.05)$ & 1.75 & $2.35\,(0.08)$ & 563.1 \\
 & 1000 & $1.72\,(0.04)$ & 1.72 & $2.26\,(0.07)$ & 443.4 \\
 & 2000 & $1.71\,(0.05)$ & 1.71 & $2.25\,(0.07)$ & 538.3 \\
\midrule
\multicolumn{6}{@{}l}{$\mathrm U(-\sqrt3,\sqrt3)$ ($\gamma_3=0.00$, $\gamma_4=-1.20$, $\gamma_6=6.9$; $1/g_2=1.00$, $1/g_3=3.33$)}\\
 & 50 & $0.92\,(0.01)$ & 0.93 & $2.18\,(0.08)$ & 2.2 \\
 & 100 & $0.96\,(0.01)$ & 0.96 & $2.78\,(0.11)$ & 1.3 \\
 & 200 & $0.99\,(0.01)$ & 0.99 & $2.99\,(0.11)$ & 0.9 \\
 & 500 & $0.99\,(0.00)$ & 0.99 & $3.18\,(0.12)$ & 0.6 \\
 & 1000 & $0.99\,(0.00)$ & 0.99 & $3.37\,(0.12)$ & 0.4 \\
 & 2000 & $1.00\,(0.00)$ & 1.00 & $3.18\,(0.12)$ & 0.3 \\
\bottomrule
\end{tabular}
\end{table}

\begin{table}[ht]
\centering
\caption{Experiment E2 ($M=2000$), laws 7--11 (continuation of Table~\ref{tab:s-e2a}). The $t_5$ law has an infinite sixth moment and violates A2.}
\label{tab:s-e2b}
\scriptsize
\setlength{\tabcolsep}{3pt}
\begin{tabular}{@{}llcccc@{}}
\toprule
Error law & $n$ & PMM$_2$ & PMM$_2^{\mathrm E}$ & PMM$_3$ & $\operatorname{sd}(\hat\gamma_6)$\\
\midrule
\multicolumn{6}{@{}l}{$\mathrm{Beta}(2,2)$ ($\gamma_3=0.00$, $\gamma_4=-0.86$, $\gamma_6=3.8$; $1/g_2=1.00$, $1/g_3=1.54$)}\\
 & 50 & $0.92\,(0.01)$ & 0.94 & $1.17\,(0.03)$ & 2.4 \\
 & 100 & $0.96\,(0.01)$ & 0.96 & $1.29\,(0.04)$ & 1.6 \\
 & 200 & $0.98\,(0.01)$ & 0.98 & $1.43\,(0.04)$ & 1.1 \\
 & 500 & $0.99\,(0.00)$ & 0.99 & $1.49\,(0.04)$ & 0.7 \\
 & 1000 & $0.99\,(0.00)$ & 0.99 & $1.51\,(0.04)$ & 0.5 \\
 & 2000 & $1.00\,(0.00)$ & 1.00 & $1.52\,(0.04)$ & 0.3 \\
\midrule
\multicolumn{6}{@{}l}{$t_{10}$ ($\gamma_3=0.00$, $\gamma_4=1.03$, $\gamma_6=12.8$; $1/g_2=1.00$, $1/g_3=1.04$)}\\
 & 50 & $0.98\,(0.01)$ & 0.99 & $0.90\,(0.02)$ & 8.9 \\
 & 100 & $0.99\,(0.01)$ & 0.99 & $0.98\,(0.01)$ & 21.1 \\
 & 200 & $1.01\,(0.01)$ & 1.01 & $1.05\,(0.01)$ & 24.5 \\
 & 500 & $1.00\,(0.01)$ & 1.00 & $1.05\,(0.01)$ & 50.6 \\
 & 1000 & $1.00\,(0.00)$ & 1.00 & $1.03\,(0.01)$ & 63.4 \\
 & 2000 & $1.00\,(0.00)$ & 1.00 & $1.04\,(0.01)$ & 31.2 \\
\midrule
\multicolumn{6}{@{}l}{$t_7$ ($\gamma_3=0.00$, $\gamma_4=2.00$, $\gamma_6=56.3$; $1/g_2=1.00$, $1/g_3=1.05$)}\\
 & 50 & $1.02\,(0.02)$ & 1.04 & $0.93\,(0.02)$ & 18.9 \\
 & 100 & $1.02\,(0.01)$ & 1.02 & $1.04\,(0.02)$ & 41.0 \\
 & 200 & $1.02\,(0.01)$ & 1.02 & $1.10\,(0.01)$ & 260.6 \\
 & 500 & $1.02\,(0.01)$ & 1.02 & $1.10\,(0.01)$ & 115.3 \\
 & 1000 & $1.00\,(0.01)$ & 1.00 & $1.08\,(0.01)$ & 156.2 \\
 & 2000 & $1.01\,(0.00)$ & 1.01 & $1.08\,(0.01)$ & 142.0 \\
\midrule
\multicolumn{6}{@{}l}{$t_5$ ($\gamma_3=0.00$, $\gamma_4=5.50$, $\gamma_6=\infty$; $1/g_2=1.00$, $1/g_3$ undefined)}\\
\multicolumn{6}{@{}l}{\quad(sample values from $4\times10^6$ draws; the population sixth moment is infinite)}\\
 & 50 & $1.06\,(0.02)$ & 1.07 & $1.02\,(0.03)$ & 40.6 \\
 & 100 & $1.05\,(0.02)$ & 1.05 & $1.14\,(0.02)$ & 123.8 \\
 & 200 & $1.04\,(0.01)$ & 1.04 & $1.18\,(0.02)$ & 252.3 \\
 & 500 & $1.02\,(0.01)$ & 1.02 & $1.16\,(0.02)$ & 1224.4 \\
 & 1000 & $1.02\,(0.01)$ & 1.02 & $1.15\,(0.02)$ & 997.1 \\
 & 2000 & $1.01\,(0.01)$ & 1.01 & $1.11\,(0.02)$ & 8187.1 \\
\midrule
\multicolumn{6}{@{}l}{$\mathcal N(0,1)$ ($\gamma_3=0.00$, $\gamma_4=0.00$, $\gamma_6=0.0$; $1/g_2=1.00$, $1/g_3=1.00$)}\\
 & 50 & $0.95\,(0.01)$ & 0.96 & $0.82\,(0.02)$ & 3.2 \\
 & 100 & $0.98\,(0.01)$ & 0.98 & $0.93\,(0.01)$ & 2.2 \\
 & 200 & $1.00\,(0.01)$ & 1.00 & $0.98\,(0.01)$ & 1.8 \\
 & 500 & $1.00\,(0.00)$ & 1.00 & $0.99\,(0.01)$ & 1.2 \\
 & 1000 & $1.00\,(0.00)$ & 1.00 & $1.00\,(0.00)$ & 0.8 \\
 & 2000 & $1.00\,(0.00)$ & 1.00 & $0.99\,(0.00)$ & 0.6 \\
\bottomrule
\end{tabular}
\end{table}

\subsection{Revision experiment E3: stabilized plug-ins}

\begin{sidewaystable}
\centering
\caption{Experiment E3 ($M=1000$): \emph{feasible} $\RE$ over OLS with paired bootstrap standard errors for the plain degree-three plug-in, the oracle-coefficient estimator (population $h$, same solver), and three stabilized plug-ins: scalar Tikhonov ridge $\hat h=(\widehat F_3+\lambda I)^{-1}\hat b$ with $\lambda=0.1\operatorname{tr}(\widehat F_3)/3$; the $\hat m_6$-ridge that adds $0.5\,\widehat F_{33}$ to the $(3,3)$ entry only; and the James--Stein-type shrinkage of $\hat m_6$ toward its value under $\kappa_6=0$. The last column is $\operatorname{sd}(\hat\gamma_6)$. The pretest-selected estimator is in Table~5 of the main paper. Large standard errors in the shrinkage column reflect a mean squared error dominated by a few divergent replications.}
\label{tab:s-e3}
\scriptsize
\begin{tabular}{@{}llcccccc@{}}
\toprule
Error & $n$ & PMM$_3$ plain & oracle $h$ & scalar ridge & $\hat m_6$-ridge & $\hat m_6$-shrinkage & $\operatorname{sd}(\hat\gamma_6)$\\
\midrule
$\chi^2(3)$
 & 50 & $2.67\,(0.16)$ & $2.35\,(0.14)$ & $0.27\,(0.02)$ & $2.21\,(0.12)$ & $0.23\,(0.70)$ & 22.5 \\
 & 100 & $3.09\,(0.16)$ & $2.61\,(0.12)$ & $0.18\,(0.01)$ & $2.32\,(0.10)$ & $0.36\,(0.19)$ & 65.6 \\
 & 200 & $2.79\,(0.15)$ & $2.31\,(0.11)$ & $0.13\,(0.01)$ & $2.09\,(0.09)$ & $0.01\,(0.15)$ & 162.7 \\
 & 500 & $2.69\,(0.13)$ & $2.49\,(0.12)$ & $0.11\,(0.01)$ & $2.20\,(0.09)$ & $0.03\,(0.15)$ & 86.4 \\
 & 1000 & $2.90\,(0.14)$ & $2.70\,(0.13)$ & $0.11\,(0.01)$ & $2.30\,(0.10)$ & $0.00\,(0.01)$ & 79.1 \\
\midrule
$\mathrm{Gamma}(2,1)$
 & 50 & $2.52\,(0.13)$ & $2.26\,(0.11)$ & $0.68\,(0.05)$ & $2.03\,(0.09)$ & $0.71\,(0.60)$ & 18.2 \\
 & 100 & $2.42\,(0.12)$ & $2.18\,(0.10)$ & $0.52\,(0.04)$ & $1.95\,(0.08)$ & $0.45\,(0.60)$ & 69.1 \\
 & 200 & $2.33\,(0.12)$ & $2.17\,(0.11)$ & $0.41\,(0.03)$ & $2.00\,(0.09)$ & $0.15\,(0.41)$ & 57.1 \\
 & 500 & $2.38\,(0.12)$ & $2.18\,(0.10)$ & $0.36\,(0.03)$ & $1.91\,(0.08)$ & $0.00\,(0.03)$ & 98.7 \\
 & 1000 & $2.18\,(0.10)$ & $2.04\,(0.10)$ & $0.41\,(0.02)$ & $1.75\,(0.07)$ & $0.01\,(0.04)$ & 57.5 \\
\midrule
$\mathrm U(-\sqrt3,\sqrt3)$
 & 50 & $2.20\,(0.12)$ & $2.10\,(0.14)$ & $2.17\,(0.09)$ & $1.55\,(0.04)$ & $1.86\,(0.13)$ & 2.2 \\
 & 100 & $3.07\,(0.17)$ & $2.98\,(0.18)$ & $2.76\,(0.11)$ & $1.93\,(0.05)$ & $3.12\,(0.19)$ & 1.4 \\
 & 200 & $2.89\,(0.16)$ & $2.86\,(0.16)$ & $2.66\,(0.11)$ & $1.98\,(0.05)$ & $2.90\,(0.16)$ & 0.9 \\
 & 500 & $3.32\,(0.17)$ & $3.38\,(0.18)$ & $2.91\,(0.11)$ & $2.08\,(0.04)$ & $3.33\,(0.17)$ & 0.5 \\
 & 1000 & $3.60\,(0.19)$ & $3.61\,(0.19)$ & $3.06\,(0.12)$ & $2.15\,(0.05)$ & $3.61\,(0.19)$ & 0.4 \\
\midrule
$\mathcal N(0,1)$
 & 50 & $0.82\,(0.02)$ & $1.00\,(0.00)$ & $0.84\,(0.02)$ & $0.93\,(0.01)$ & $0.55\,(0.03)$ & 3.4 \\
 & 100 & $0.90\,(0.02)$ & $1.00\,(0.00)$ & $0.84\,(0.02)$ & $0.94\,(0.01)$ & $0.80\,(0.03)$ & 2.1 \\
 & 200 & $0.95\,(0.01)$ & $1.00\,(0.00)$ & $0.84\,(0.02)$ & $0.97\,(0.01)$ & $0.92\,(0.02)$ & 1.5 \\
 & 500 & $0.98\,(0.01)$ & $1.00\,(0.00)$ & $0.89\,(0.02)$ & $0.99\,(0.01)$ & $0.98\,(0.01)$ & 1.3 \\
 & 1000 & $0.99\,(0.01)$ & $1.00\,(0.00)$ & $0.87\,(0.02)$ & $1.00\,(0.01)$ & $0.99\,(0.01)$ & 0.8 \\
\bottomrule
\end{tabular}
\end{sidewaystable}

\subsection{Heteroskedasticity}

\begin{table}[ht]
\centering
\caption{Heteroskedasticity Monte Carlo ($M=2000$; Section~5.6 of the main paper). Columns OLS--WLS are $\RE=\mathrm{MSE}_{\OLS}/\mathrm{MSE}$; the last column is the PMM$_2$ slope bias (all other estimators have $|\mathrm{bias}|<0.005$). PMM$_2$(unc.) uses the unconditional body $F_2$; SLS(cond.) uses the fitted conditional variance model $\exp(\hat\tau_0+\hat\tau_1x+\hat\tau_2x^2)$; WLS(orac.) is weighted least squares with the true weights $1/\sigma^2(x_i)$. Design~A (symmetric, $\operatorname{Cov}(x,\sigma^2(x))=0$): efficiency separation only. Design~B (asymmetric, $\operatorname{Cov}(x,\sigma^2(x))\neq0$): the unconditional-body PMM$_2$ slope is inconsistent.}
\label{tab:s-hetero}
\small
\begin{tabular}{@{}llccccc@{}}
\toprule
design & $n$ & OLS & PMM$_2$(unc.) & SLS(cond.) & WLS(orac.) & PMM$_2$ bias\\
\midrule
A (sym.)  & 100  & 1.000 & 0.965 & 1.013 & 1.085 & $+0.004$\\
          & 500  & 1.000 & 0.990 & 1.063 & 1.089 & $+0.002$\\
          & 2000 & 1.000 & 1.001 & 1.091 & 1.103 & $-0.001$\\
\midrule
B (asym.) & 100  & 1.000 & 0.415 & 1.250 & 1.291 & $-0.109$\\
          & 500  & 1.000 & 0.112 & 1.257 & 1.278 & $-0.100$\\
          & 2000 & 1.000 & 0.035 & 1.309 & 1.302 & $-0.094$\\
\bottomrule
\end{tabular}
\end{table}

\subsection{Verification of the simulation code}

The simulation study of the original submission was checked along five lines: the correctness of the feasible-optimal SLS implementation, the significance of the equivalence under a paired bootstrap, the closed-form $g_2$ algebra, the framing of the symmetric experiment, and the overall design; the scripts performing these checks are part of the code supplement. The one issue found, a seed collision across cells, left all within-cell paired comparisons unaffected and was corrected. Two by-products reinforce the theoretical claims: an independent asymptotic sandwich computation for the $\chi^2(3)$ case returned $\RE=1.798$ against the theoretical $1.800$, and a direct calculation of $\sqrt n\,\widehat{\operatorname{sd}}(\hat\beta_1^{\mathrm{PMM_2}}-\hat\beta_1^{\mathrm{SLS}})\to0$ corroborates the strong PMM$_2$--SLS equivalence. The checks also re-established the implementation point of Section~3.4 of the main paper: the degree-two gain materializes only with the regression-form residual $y^2-\mu^2-\sigma^2$ and disappears under the finite-sample-degenerate $e^2-\sigma^2$ form. For the revision experiments the general degree-three solver (\texttt{common.R::pmm\_reg}) was checked against \textsf{EstemPMM} on the symmetric laws, where the two agree (Table~\ref{tab:s-e1}), and against the population coefficients (oracle column of Table~\ref{tab:s-e3}).

\section{Data examples: details}
\label{sec:s6}

\subsection{\emph{cars} and \emph{faithful}}

Both data sets ship with R (\texttt{datasets::cars}, \texttt{datasets::faithful}). \emph{cars} \citep{ezekiel1930methods}: stopping distance (ft) on speed (mph), $n=50$; \emph{faithful} \citep{azzalini1990look}: eruption duration (min) on waiting time (min), $n=272$. For each set we fit OLS, compute the residual central moments through order six and the standardized cumulants, the Breusch--Pagan test of \citet{breusch1979simple} as implemented in \texttt{lmtest::bptest}, and the plug-in $1/\hat g_2$, $1/\hat g_3$ from the residual moments. Point estimates of the slope are OLS, the general degree-two and degree-three plug-ins, the \textsf{EstemPMM} degree-two routine, and feasible SLS; their standard errors are paired case-bootstrap standard errors over $B=2000$ resamples of $(y_i,x_i)$ pairs (seed 101). The residual-calibrated Monte Carlo resamples errors from the centered OLS residuals and the regressor from the observed $x$, generates $y=\hat\beta_0+\hat\beta_1x+e$, and runs $M=2000$ replications (seed 202) at $n=n_{\text{data}}$ and at $n=1000$; $\RE$ is measured against $\hat\beta_1$. Table~\ref{tab:s-e4} gives the full results; the script is \texttt{e4\_real\_data.R}.

\begin{table}[ht]
\centering
\caption{Data examples: slope estimates with paired case-bootstrap standard errors ($B=2000$), and residual-calibrated Monte Carlo $\RE$ over OLS with paired bootstrap standard errors ($M=2000$). PMM$_2^{\mathrm E}$ is the \textsf{EstemPMM} degree-two routine.}
\label{tab:s-e4}
\scriptsize
\setlength{\tabcolsep}{2pt}
\begin{tabular}{@{}lcccccc@{}}
\toprule
& \multicolumn{3}{c}{\emph{cars} ($n=50$)} & \multicolumn{3}{c}{\emph{faithful} ($n=272$)}\\
\cmidrule(lr){2-4}\cmidrule(lr){5-7}
estimator & slope (SE) & $\RE$, $n=50$ & $\RE$, $n=1000$ & slope (SE) & $\RE$, $n=272$ & $\RE$, $n=1000$\\
\midrule
OLS & $3.932\,(0.409)$ & 1 & 1 & $0.0756\,(0.0019)$ & 1 & 1\\
PMM$_2$ & $3.437\,(0.481)$ & $1.304\,(0.031)$ & $1.411\,(0.031)$ & $0.0759\,(0.0020)$ & $0.990\,(0.006)$ & $1.007\,(0.006)$\\
PMM$_3$ & $3.233\,(0.504)$ & $1.212\,(0.032)$ & $1.458\,(0.035)$ & $0.0760\,(0.0020)$ & $1.040\,(0.015)$ & $1.086\,(0.015)$\\
PMM$_2^{\mathrm E}$ & $3.482\,(0.408)$ & $1.288\,(0.030)$ & $1.410\,(0.031)$ & $0.0759\,(0.0020)$ & $0.991\,(0.005)$ & $1.007\,(0.006)$\\
SLS & $3.599\,(0.398)$ & $1.275\,(0.030)$ & $1.405\,(0.031)$ & $0.0756\,(0.0019)$ & $0.988\,(0.006)$ & $1.003\,(0.005)$\\
\bottomrule
\end{tabular}
\end{table}

\subsection{The FORGE 58-32 innovation law}

\emph{Source.} Utah FORGE well 58-32, processed 1-ft Pason drilling log, file \path{Well_58-32_processed_pason_log.csv} of the Geothermal Data Repository submission ``Utah FORGE: Drilling Data for Student Competition'' \citep{podgorney2018forge}, licence CC~BY~4.0. Columns used: depth (ft), rate of penetration (ft/h, 1-ft average), weight on bit (klbf), rotary speed (rpm).

\emph{Why not a regression sample.} Rows with non-positive or missing rate of penetration, weight on bit or rotary speed are dropped. Within a depth window the residuals of an OLS fit of rate of penetration on weight on bit are strongly serially dependent (lag-one autocorrelation between $0.35$ and $0.81$, Table~\ref{tab:s-windows}), and rotary speed changes between windows, so the data violate A1 as a regression sample and are used only as a source of an empirical error law.

\emph{Pipeline.} Rotary speed is binned to the nearest 10~rpm. For each 500-ft window with at least 100 rows, the dominant rotary-speed bin is selected and retained if it holds at least 80 rows; OLS of rate of penetration on weight on bit is fitted in that subset; an AR(1) model without mean (\texttt{stats::arima}) is fitted to the residuals; the innovations are standardized to zero mean and unit variance and pooled over windows. Fourteen windows survive (depths 500--7500~ft), giving $n=3266$ pooled innovations. A second law removes innovations beyond a Tukey fence at three inter-quartile ranges ($0.9\%$ of the sample, $n=3237$). Table~\ref{tab:s-windows} lists the windows; the pooled and fenced cumulants and the plug-in $1/\hat g_2$, $1/\hat g_3$ are in Table~9 of the main paper.

\emph{Calibrated Monte Carlo.} Errors are resampled with replacement from the pooled (or fenced) law and inserted into the design of the main paper, $y=2+1.5x+e$, $x\sim\mathrm U(0,5)$, at $n\in\{100,250,1000\}$ with $M=2000$ replications (seed $20260915+n$); OLS, the general degree-two and degree-three plug-ins, feasible SLS and the six-condition GMM$_3$ are compared. The script is \texttt{e5\_forge\_calibrated.R}.

\begin{table}[ht]
\centering
\caption{FORGE 58-32: the fourteen 500-ft windows behind the pooled innovation law. rpm is the dominant rotary-speed bin; $n$ the number of rows in it; $\hat\phi_1$ the fitted AR(1) coefficient; acf$_1$ the lag-one residual autocorrelation before and after prewhitening; $\hat\gamma_3,\hat\gamma_4,\hat\gamma_6$ the standardized cumulants of the window's innovations.}
\label{tab:s-windows}
\footnotesize
\begin{tabular}{@{}lcccccccc@{}}
\toprule
window (ft) & rpm & $n$ & $\hat\phi_1$ & acf$_1$ before & acf$_1$ after & $\hat\gamma_3$ & $\hat\gamma_4$ & $\hat\gamma_6$\\
\midrule
500--1000 & 80 & 191 & 0.47 & 0.47 & $-0.14$ & $-1.59$ & 3.8 & $-22$ \\
1000--1500 & 90 & 190 & 0.49 & 0.48 & $-0.12$ & $-0.48$ & 3.4 & 14 \\
1500--2000 & 80 & 173 & 0.70 & 0.70 & $-0.22$ & 0.60 & 1.4 & $-9$ \\
2000--2500 & 60 & 142 & 0.36 & 0.35 & $-0.06$ & $-0.08$ & $-0.3$ & 0 \\
2500--3000 & 50 & 336 & 0.48 & 0.48 & $-0.10$ & $-0.11$ & 0.0 & $-2$ \\
3000--3500 & 50 & 201 & 0.62 & 0.62 & $-0.21$ & 0.81 & 3.5 & 33 \\
3500--4000 & 50 & 223 & 0.63 & 0.63 & $-0.08$ & 0.37 & 2.9 & 22 \\
4000--4500 & 60 & 234 & 0.72 & 0.72 & 0.07 & 0.97 & 12.5 & 272 \\
4500--5000 & 60 & 125 & 0.55 & 0.54 & $-0.13$ & 0.41 & 1.3 & $-1$ \\
5000--5500 & 40 & 239 & 0.56 & 0.57 & $-0.15$ & 2.77 & 18.5 & 771 \\
5500--6000 & 40 & 210 & 0.68 & 0.61 & $-0.09$ & 2.26 & 13.9 & 472 \\
6000--6500 & 40 & 389 & 0.81 & 0.80 & $-0.30$ & 0.11 & 4.4 & 47 \\
6500--7000 & 40 & 341 & 0.62 & 0.61 & $-0.20$ & 0.94 & 6.4 & 73 \\
7000--7500 & 40 & 272 & 0.60 & 0.60 & $-0.04$ & 3.79 & 35.7 & 2820 \\
\bottomrule
\end{tabular}
\end{table}